%% file: arxiv-version.tex
\documentclass[sigplan,10pt]{acmart}
\renewcommand\footnotetextcopyrightpermission[1]{}

\usepackage{tikz}
\usepackage{amsmath}

\usepackage{filecontents}

\usepackage[labelformat=simple]{subcaption}
\usepackage{enumitem}
\usepackage{tcolorbox}
\usepackage{algorithm}
\usepackage[noend]{algpseudocode}
\usepackage{listings}
\usepackage[xcolor]{changebar}
\usepackage{soul}
\usepackage{etoolbox}
\newbool{publicversion}
\usepackage{caption}
\usepackage{lineno}
\usepackage{lipsum}
\usepackage{amsthm}
\usepackage{thmtools, thm-restate}
\usepackage{graphicx}

\usepackage{amsmath, amssymb}
\usepackage[bottom]{footmisc}
\usepackage[subtle]{savetrees}

\newtheoremstyle{proofSketchStyle}%
      {\topsep}{\topsep}%
      {\upshape}%          Body font
      {}%                 Indent amount (empty = no indent, \parindent = para indent)
      {\itshape}%        Head font
      {.}%                Punctuation after head
      {5pt plus 1pt minus 1pt}% Space after head (\thm@space)
      {\thmname{#1}\thmnumber{ #2}\thmnote{ (#3)}}% Head spec

\newtheorem{lemma}{Lemma}

\theoremstyle{proofSketchStyle}
\newtheorem*{proofSketch}{Proof Sketch}

\theoremstyle{definition}
\newtheorem{definition}{Definition}

\algrenewcommand{\algorithmiccomment}[1]{\hfill\ensuremath{//} #1}
\makeatletter
\newcommand{\LeftComment}[1]{%
  \Statex \hskip\ALG@tlm // #1%
}
\makeatother
\algrenewcommand{\Call}[2]{\textbf{call}\ \textsc{#1}(#2)}

\newcommand{\shock}[0]{shock}
\newcommand{\dsl}[0]{Nyx}
\newcommand{\mftcm}[0]{$\texttt{CM}_{\texttt{mft}}$}
\newcommand{\phenomenological}[0]{phenomenological}
\newcommand{\ontological}[0]{analytical}

\setbool{publicversion}{true}

\definecolor{codegreen}{rgb}{0,0.6,0}
\definecolor{codegray}{rgb}{0.5,0.5,0.5}
\definecolor{codepurple}{rgb}{0.58,0,0.82}
\definecolor{backcolour}{rgb}{1,1,1}

\newif\ifshowcomments
\showcommentstrue

\ifbool{publicversion}{
  \newcommand{\grumbler}[2]{}
  \newcommand{\grumblermargin}[2]{}
  
  \newcommand{\pagelimit}[1]{}
  \newcommand{\todos}[1]{}
  }

\lstdefinestyle{mystyle}{
    backgroundcolor=\color{backcolour},   
    commentstyle=\color{codegreen},
    keywordstyle=\color{magenta},
    numberstyle=\tiny\color{codegray},
    stringstyle=\color{codepurple},
    basicstyle=\ttfamily\scriptsize,
    breakatwhitespace=false,         
    breaklines=true,                 
    captionpos=b,                    
    keepspaces=true,                 
    numbers=none,
    frame=none,
    numbersep=0pt,
    showspaces=false,                
    showstringspaces=false,
    showtabs=false,                  
    tabsize=2,
    gobble=8
}

\begin{document}
\emergencystretch 3em

\title{Characterizing Metastable Faults and Failures}
%\author{Submission \#1979}
% \author{
% {Ali Farahbakhsh}\\
% Cornell University
% \and
% Qingjie Lu\\
% University of Pennsylvania
% \and
% Lorenzo Alvisi\\
% Cornell University
% \and
% {Andreas Haeberlen}\\
% University of Pennsylvania
% \and
% Robbert Van Renesse\\
% Cornell University
% }

\author{
Ali Farahbakhsh$^\dagger$, Qingjie Lu$^*$, Lorenzo Alvisi$^\dagger$, Andreas Haeberlen$^*$, Robbert Van Renesse$^\dagger$\\
$^\dagger$Cornell University, $^*$University of Pennsylvania
}

% \author[1]{Ali Farahbakhsh}
% \author[2]{Qingjie Lu}
% % \author[1]{Lorenzo Alvisi}
% % \author[2]{Andreas Haeberlen}
% % \author[1]{Robbert Van Renesse}
% \affil[1]{Cornell University}
% \affil[2]{University of Pennsylvania}

\begin{abstract}
\input{abstract}
\end{abstract}

\settopmatter{printfolios=true}
\maketitle
\pagestyle{plain}

%\onecolumn

%\vspace{-5em}

% \begin{abstract}
% \input{abstract}
% \end{abstract}

\input{introduction}
\input{case-1}
\input{model}
\input{faults}
\input{failures}
\input{macroscopics}
\input{case-2}
\input{related}
\input{conclusion}

\bibliographystyle{abbrv} 
\begin{small}
\bibliography{main}
\end{small}

\clearpage

\appendix
\input{appendix}
\end{document}

%% file: abstract.tex
Metastable failures are hard to detect, prevent, and mitigate.
During a metastable failure, a system exhibits self-sustaining bad behavior even in the absence of adversarial conditions. Prior work focuses on symptoms and has 
portrayed metastable failures as instances of self-sustaining overload. This characterization leaves the underlying failure causes and dynamics unknown, and does not account for metastable failures that do not manifest as overload.

We present the first causal  characterization of metastable failures by identifying their origin in {\em metastable faults}, {\em i.e.}, structural destabilizing cycles of interaction among systems components that, in isolation, are stabilizing. Metastable failures arise when scheduling decisions let these destabilizing interactions gain the upper hand over the individual components' stabilizing tendencies. 
We then derive a methodology to predict metastable failures, and to build metastable-fault-tolerant (MFT) systems.
We apply our methodology to three case studies, showcasing the generality of our results.

%We show that they happen when components in a composition destabilize each other while trying to stabilize locally 
%We introduce a separation between metastable faults and failures: faults are structural destabilizing interaction cycles among components.
%Failures, in turn, happen when the scheduler prioritizes the fault over stabilizing tendencies, leading to perpetual instability.
%Our formalization reveals that metastable faults are necessary to have metastable failures.
%We then derive a methodology to predict metastable failures, and to build metastable-fault-tolerant systems; systems that do not exhibit failures despite having a fault.
%The key is to augment the system with the right scheduler.
%We apply our methodology to three case studies, showcasing the generality of our results.

%% file: introduction.tex
\section{Introduction}
%\la{I have added the "changebars" command, and a boolean "publicversion" (in main) which, if set to true, makes all the changebars disappear.}

\emph{Metastable failures}~\cite{understanding-vicious-cycles-khan2011, qian-2023-viciouscycles, ford2012icebergs, link-imbalance-bronson2014, floyd1993synchronization, GGL1, AWS1, AWS2, AWS3, SPF1} have drawn significant attention in the systems community as a particularly challenging failure mode---one in which a system exhibits {\em self-sustaining bad behavior after a finite \shock{}}~\cite{bronson-2021-metastability, huang-2022-metastable, analyzing-metastable-failures, hotnets2025, habibi-2023-msfmodel, blueprint-anand2023, brooker-2023-socc-keynote}.
For instance, a temporary surge in client load might lead to persistently low goodput even after the load goes back to normal.
The system does not control the \shock{}; however, the ensuing bad behavior is due solely to how the system works.
Such failures are costly~\cite{huang-2022-metastable}, hard to detect or predict~\cite{link-imbalance-bronson2014}, and difficult to mitigate~\cite{GGL1, AWS1, AWS2, AWS3}, making it critical to understand their underlying causes.

The prevailing characterization of metastable failures is mostly \emph{\phenomenological{}}: it starts from symptoms---high latency and low goodput under normal conditions---and equates metastability with self-sustaining overload~\cite{bronson-2021-metastability, huang-2022-metastable}.
It then identifies the culprit as either an internal load amplification or capacity degradation~\cite{huang-2022-metastable}.
By focusing on symptoms, this approach fails to fully explain the underlying causes that give rise to metastable failures, leaving us vulnerable in two ways: ($i$) when these failures occur, we lack targeted remedies---reports repeatedly show that only drastic measures such as restarting the system or disabling communication can halt a live failure~\cite{GGL1, AWS1, AWS2, AWS3, SPF1}; and ($ii$) we risk overlooking metastable failures that do not manifest with the same symptoms.

To move beyond these limitations, we return to a fundamental principle of fault tolerance: {\em every failure stems from a fault}.
%We rely instead on the fault tolerance credo: if there is a failure, then there is a fault.
We present the first \ontological{} characterization of metastable failures that reveals the structural faults that cause them---henceforth called metastable faults.
\emph{Metastable faults are sins of composition}: they
capture circular destabilizing interactions among individually stabilizing components.
Each component responds to a \shock{} by trying to stabilize locally, but in doing so destabilizes others, coupling them in perpetual destabilization and producing self-sustaining bad behavior.
This characterization also explains the usual symptoms of high latency and low goodput: resources that in the stable state would be spent in useful work are instead continuously wasted in futile attempts to stabilize.

Building on our experience studying various metastable failures, we outline a methodology for predicting such failures and designing \emph{metastable-fault-tolerant} (MFT) systems.
Our methodology revolves around a proof outline for showing that a system is MFT,~{\em i.e.}, that the system avoids metastable failures despite harboring a metastable fault.
We also introduce a domain-specific language (DSL), \dsl{}, to ($i$) reproduce said failures {\em in vitro} and ($ii$) assist designers in applying the methodology.
\dsl{} comes with a toolkit including an interpreter and a tool to visualize metastable failures.%\af{More on \dsl{}}.

Our methodology relies on a key insight:
metastable failures emerge when poor scheduling favors the fault over stabilizing interactions.
In the absence of stabilizing interactions, perpetual destabilization can be trivially inevitable.
%That composing a switch that turns a light off with one that turns it on yields an oscillating bulb is unsurprising---the switches were never meant to be composed.
Metastable failures occur instead when the components \emph{can} help each other stabilize, but any stabilization is interrupted by some destabilizing event.
The reason is a fatal coupling between destabilization and poor scheduling: the scheduler favors the fault's destabilizing tendencies, pushing the system away from stability, creating in turn more destabilizing interactions, prompting the system to make more bad scheduling decisions---{\em ad infinitum}.

Based on this insight, our methodology augments compositions of stabilizing components with scheduling decisions that make them MFT.
These decisions must:
($i$) schedule destabilizing interactions tentatively, and
($ii$) defer destabilizing interactions until stabilizing ones have achieved global stability.
The latter ensures eventual global stability; the former ensures that, once achieved, stability is never lost.

We prove metastable fault tolerance with respect to specific adversaries.
Faults are notoriously difficult to locate~\cite{link-imbalance-bronson2014, understanding-vicious-cycles-khan2011}, and often require post-mortem analysis.
However, these faults are triggered by an adversary.
Therefore, proving stabilization despite an adversary imposing \shock{}s on the system establishes tolerance against all faults that the adversary can trigger---a single proof covers a wide range of faults.
We envision research on metastable faults and failures to resemble security research: new adversaries will expose new vulnerabilities.

% The second insight is that metastable failures can be visualized even without knowing of the faults causing them.
% Following similar work~\cite{analyzing-metastable-failures}, we visualize metastable failures using vector fields: graphs demonstrating a system's dynamic behavior starting from various points of the state space.
% Given the tug of war between destabilizing and stabilizing interactions, metastable failures manifest as two coexisting tendencies within the vector field: one pulling towards good stable states, and the other pushing away from it.
% Our characterization and methodology makes it easier for the designer to pick suitable axes for the vector field, revealing the mentioned tendencies.
% The \dsl{} toolkit facilitates generating vector fields for systems expressed using the \dsl{} DSL, thereby helping designers predict metastable failures.

We have applied our methodology to three case studies: ($i$) the retry storm~\cite{GGL1, AWS1, AWS2, AWS3, retry-anti-pattern}, which serves as the running example for our characterization, ($ii$) a novel case study from a major commercial gaming platform with hundreds of millions of users, where metastability manifests as oscillation, and ($iii$) the look-aside cache incident~\cite{bronson-2021-metastability, huang-2022-metastable}, used to demonstrate the generality of our approach.

For the second case study, we use an internal post-mortem of an incident involving a cluster manager in one of the platform's data centers. We design a new cluster manager and prove it MFT. The fault in this incident is subtle and hard to locate without post-mortem insight, yet we show that fault tolerance is achievable without prior knowledge of the fault.
The third case study instead illustrates that sometimes one can simply remove the fault, thereby eliminating also the failure.

Our methodology is a modest first step towards developing and deploying MFT software in production.
Metastable faults span the entire stack~\cite{link-imbalance-bronson2014, understanding-vicious-cycles-khan2011}, and the vast scale of production systems makes locating them even more challenging.
As a result, while the community has developed robust toolboxes to prevent,  detect, and mitigate other modes of failure,~{\em e.g.}, deadlocks~\cite{Dijkstra-banker, havender68, wang2008gadara}, we lack an equivalent toolbox for metastable failures.
Our characterization lays the foundation  for such a pursuit by identifying the underlying principles, and our methodology demonstrates the practical utility of this characterization.

In short, we make the following contributions:
\begin{itemize}
    \item We introduce metastable faults, and show how together with poor scheduling they lead to metastable failures;
    \item we present a methodology for predicting metastable failures and designing MFT systems; and
    \item we introduce \dsl{}, a DSL to reproduce metastable failures in vitro.
\end{itemize}
% ($i$) we introduce metastable faults, a structural cause of metastable failures

The paper is structured as follows.
We first use the retry storm incident as a motivating running example (\S\ref{sec:background}) to explain, given a system model (\S\ref{sec:model}), our formal characterization of metastable faults (\S\ref{sec:faults}) and failures (\S\ref{sec:failures}).
We then present our methodology (\S\ref{sec:methodology}), and apply it to the privately reported incident (\S\ref{sec:MFT}).
Finally, we discuss related work (\S\ref{sec:related-work}) and present our conclusions (\S\ref{sec:conclusion}).
Space constraints prevent us from including the discussion of the look-aside cache incident here; it is provided in the supplementary materials (Section~\ref{sec:cache-incident}).

% For its canonical image, we ask the reader to conduct a thought experiment: what happens if we, starting from the edge of a bowl, drop several magnets inside the bowl.
% Is it guaranteed that all of the magnets will stabilize at the bottom of the bowl?

%% file: case-1.tex
\section{Background and Motivation}\label{sec:background}

Among reported metastability incidents, one stands out: the {\em retry storm}.
It has been observed repeatedly in practice~\cite{GGL1, AWS1, AWS2, AWS3, retry-anti-pattern}, and much of the literature focuses on retry storms~\cite{huang-2022-metastable, habibi-2023-msfmodel, analyzing-metastable-failures}.
%, often equating metastability with “self-sustaining congestive collapse”~\cite{analyzing-metastable-failures}.
This makes it an ideal example to ($i$) illustrate the limits of today's \phenomenological{} approach and ($ii$) serve as a running example in our \ontological{} characterization (\S\ref{sec:faults}-\ref{sec:failures}).

%Among the reported incidents on metastability, one serves a special purpose: \emph{the retry storm}.
%It is the poster child of metastable failures, and has been observed repeatedly in practice~\cite{}.
%Much of the existing literature focuses on the retry storm~\cite{huang-2022-metastable, habibi-2023-msfmodel, analyzing-metastable-failures}, which has led to a definition that equates metastability with ``self-sustaining congestive collapse''~\cite{analyzing-metastable-failures}.
%It thus serves as an ideal candidate to ($i$) explain the shortcomings of the \phenomenological{} approach to studying metastable failures, and to ($ii$) use as a running example in our \ontological{} characterization of metastability (\S\ref{sec:faults}-\ref{sec:failures}).

%The purpose of this section is to set the stage for the characterization in Section~\ref{sec:metastability} by explaining the retry storm with the help of \dsl{}, our DSL for studying metastable failures, as it were, {\em in vitro}.
We start with a brief primer on \dsl{} (\S\ref{sec:nimbus-primer}).
We then instantiate the retry storm using \dsl{} and review it (\S\ref{sec:incident-retry-storm}).
Finally, we argue that the \phenomenological{} approach, while helpful, fails to fully explain metastable failures (\S\ref{sec:existing-approach}).

\subsection{\dsl{} Primer}\label{sec:nimbus-primer}

\dsl{} is an imperative language with a small set of abstractions that repeatedly surface when studying metastability incidents: ($i$) {\em agents}---entities in a distributed execution that send/receive messages and update local state; ($ii$) input/output {\em queues} for communication; and
($iii$) internal {\em resources} that agents use to schedule state transitions.

A \dsl{} agent defines a set of tasks, similar to threads, plus a main task scheduling and executing them based on resource availability.
\dsl{} interprets code for a group of agents and delivers requests to them during execution.
Executions proceed in lock-step: at each step, every agent runs its main task once.

\subsection{The Retry Storm}\label{sec:incident-retry-storm}

Consider a system with a fixed service capacity serving client requests, and imagine a sudden surge in client demand temporarily exceeding this capacity.
Such a \shock{} leads to congestion and increased latency, which in turn triggers client retries.
Persist this long enough and the volume of retries grows so large that, even after client demand returns below capacity, congestion prevails.
The retry loop sustains overload and high latency, creating a self-perpetuating cycle of congestion and retries.

\textbf{Instantiation.}
Figure~\ref{listing:retry} illustrates a \dsl{} model of a system prone to a retry storm.
The system has two agents: {\em server} and {\em retrier}.

The server has 35 units of a resource of type~$\texttt{CPU}$\footnote{This name is evocative, not tied to an actual CPU.}, and two tasks: $\texttt{serve}$ and $\texttt{main}$.
$\texttt{serve}$ consumes~$\texttt{CPU}$ resources to process requests from input queue~$\texttt{inq}$---declared implicitly for every agent---and sends acknowledgments to the retrier; it abstracts away service logic.
$\texttt{main}$ schedules~$\texttt{serve}$ using a~$\texttt{soft}$ policy, allowing execution until $\texttt{CPU}$ is exhausted.
\dsl{} handles resource accounting;~$\texttt{serve}$ consumes one unit of~$\texttt{CPU}$ per every request, and when~$\texttt{CPU}$ is depleted,~$\texttt{serve}$ pauses until the next step, when~$\texttt{CPU}$ is replenished.
For simplicity, we assume that~$\texttt{inq}$ has infinite capacity.
%This does not qualitatively change why and how a retry storm emerges. 

The retrier consists of three tasks:~$\texttt{manage}$,~$\texttt{retry}$, and~$\texttt{main}$.
$\texttt{manage}$ handles incoming requests---adds~$\texttt{client}$ requests to state variable~$\texttt{pending}$, forwards them to the server, and removes them upon receiving an~$\texttt{ack}$.
$\texttt{retry}$ implements a retry loop.
For each pending request, it increments its timer: upon  exceeding a threshold (every four steps), it sends a retry to the server and resets the timer.
$\texttt{main}$ schedules both~$\texttt{manage}$ and~$\texttt{retry}$ without resource constraints.
The retrier logically centralizes client behavior, behaving similar to multiple clients acting in concert.
We assume an external agent (omitted for brevity) that continuously sends \texttt{client} requests to the retrier in an open loop.

\begin{figure}[t]
        \begin{lstlisting}[language=Python, numbers=left, basicstyle=\scriptsize, mathescape]
         $\textbf{RETRIER}$:
          Task manage:
              while not inq.isEmpty():
                  $\texttt{req}$ $\leftarrow$ inq.get()
                  if $\texttt{req}$ is "client":
                      send($\texttt{req}$, server)
                      pending $\leftarrow$ pending $\cup\ \{\texttt{req}\}$
                  else if $\texttt{req}$ is "ack":
                      pending $\leftarrow$ pending $\setminus \ \{\texttt{req}\}$
          Task retry:
              for $\texttt{req}$ in pending:
                  timers[$\texttt{req}$] += 1
                  if timers[$\texttt{req}$] == 4:
                      send($\texttt{retry}$, server)
                      timers[$\texttt{req}$] = 0
          Task main:
              execute(manage)
              execute(retry)

         $\textbf{SERVER}$:
          init resources = <CPU: 35>
          Task serve consumes <CPU>:
                while not inq.isEmpty():
                    $\texttt{req}\ \leftarrow$ inq.get()
                    send($\texttt{ack}$, retrier)
          Task main:
                execute(serve, soft)
        \end{lstlisting}
        
    \caption{\dsl{} expression for the retrier and the server. The client is not shown for lack of space. The server's finite capacity is captured by its consumption of \texttt{CPU} and the \texttt{soft} execution policy.
    }
    \label{listing:retry}
\end{figure}

\textbf{Demonstration.} Figure~\ref{fig:retry-storm} shows the results of running our retry storm instantiation.
The server processes~$s=35$ requests per timestep under a nominal client load of~$r=30$, with retrier timeout~$T=4$ timesteps.
We analyze three scenarios: ($i$) {\em No \shock{}} -- the system operates normally; ($ii$) {\em Safe \shock{}} -- a transient surge in client load that does not destabilize the system; and ($iii$) {\em Unsafe \shock{}}, a surge that triggers self-sustaining congestion.
The figure reports two metrics: pending requests in the retrier (a proxy for latency), and acks received by the retrier (a proxy for goodput).

Under an unsafe \shock{}, the backlog of pending requests grows unbounded, and goodput collapses.
A safe \shock{} causes only temporary disruption and even a short-lived goodput boost due to the elevated client load.

\begin{figure}[t]
  \centering
    \subfloat[Goodput.]{\includegraphics[width=0.45\columnwidth]{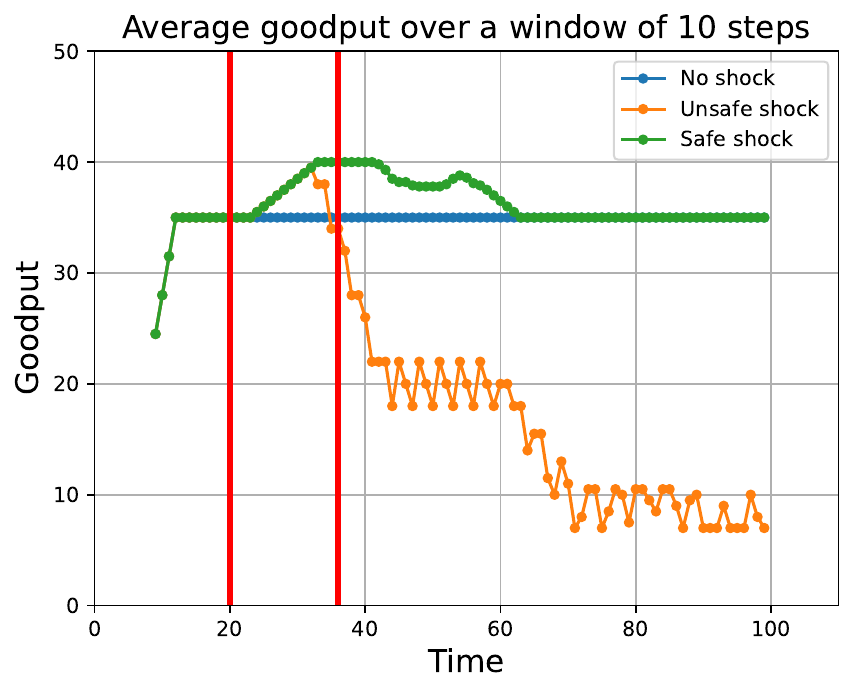}%\label{fig:retry-storm-notrigger}}
    }
    %\hspace{5mm}
    \subfloat[Pending requests.]{\includegraphics[width=0.49\columnwidth]{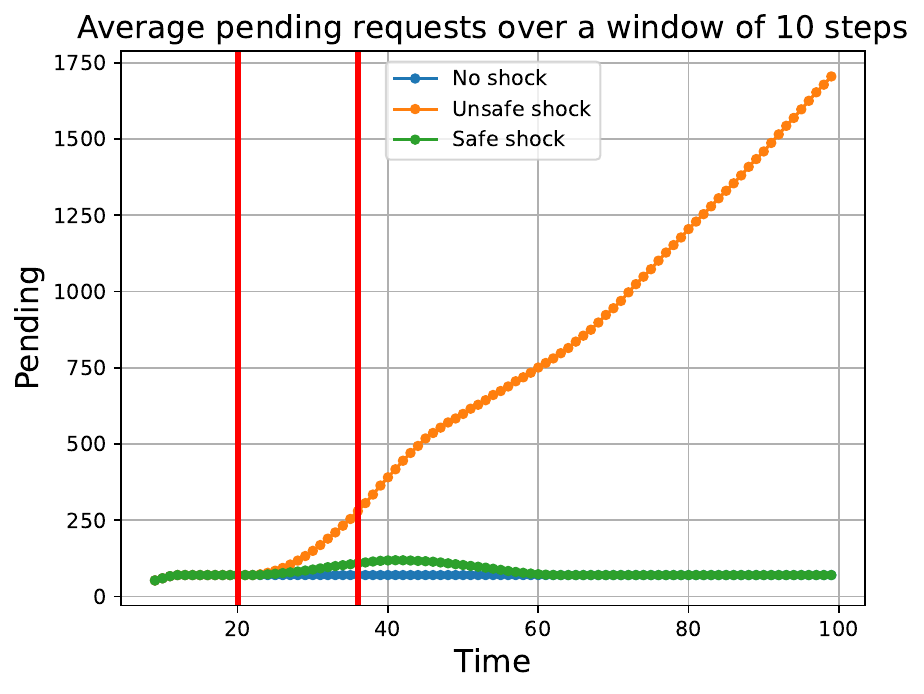}%\label{fig:retry-storm-withtrigger}}
    }
    \vspace{-2mm}\caption{\dsl{} instantiation of the retry storm for~$s=35,~r=30$, and~$T=4$. The red lines indicate when the surge in client load starts and ends.}\label{fig:retry-storm}
\end{figure}

\textbf{Review of the Current Approach.} The prevailing characterization of metastability~\cite{huang-2022-metastable} explains the retry storm as a case of {\em retry-driven workload amplification}.
It starts from how metastable failures usually manifest---high latency and low goodput under otherwise normal operating conditions---and puts the blame on either an internal amplification of load or a degradation of service capacity. These are the self-sustaining effects of a feedback loop that perpetuates the metastable failure.
The \shock{} temporarily amplifies workload or degrades capacity, long enough for the self-sustaining effects to take hold.
This framework thus defines four classes of metastable failures, based on whether the shock and sustaining effect amplify workload or degrade capacity.

% Under this characterization, the retry storm phenomenon is explained as a case of retry-driven workload amplification.
% The retry feedback loop generates enough load to perpetuate itself, keeping load above capacity for extended periods.
% The initial \shock{} may take the form of workload amplification ({\em e.g.}, a surge in client load) or capacity degradation ({\em e.g.}, a server slow-down).

\subsection{Shortcomings of the Current Approach}\label{sec:existing-approach}

% The current approach to metastability is largely phenomenological: it focuses on observable symptoms—namely, high latency and low goodput. Consequently, it remains tied to categories such as overload and congestion, culminating in definitions like “self-sustaining congestive collapse”~\cite{analyzing-metastable-failures}. This perspective suffers from three key shortcomings.

The prevailing approach focuses largely on observable symptoms, and remains tied to categories like overload and congestion, ultimately equating metastability with ``self-sustaining congestive collapse"~\cite{analyzing-metastable-failures}.
It suffers from three shortcomings.

\textbf{(1) No Causal Insight.} It ignores 
the underlying causes and internal dynamics of metastable failures.
If metastability were the flu, this approach treats the fever, not the virus.

In the retry storm (Figure~\ref{listing:retry}), once a \shock{} persists long enough, retries flood the server’s queue, crowding out original requests.
The server wastes cycles acknowledging retries for requests already served, effectively reducing its capacity and amplifying the retrier’s workload.
The root cause is not just the retry feedback loop and its amplification of load; it is rather a circular interplay between workload amplification and capacity degradation, shaped by ($i$) the agents’ functional composition and ($ii$) their scheduling policies, which contribute to workload amplification or capacity degradation.

\textbf{(2) Inadequate Remedies.} Without causal insight, attempted remedies are often suboptimal.
Shedding client load is the standard response, yet our analysis shows the real issue lies in scheduling.
The retrier retries aggressively every~$T$ timesteps,\footnote{Finite retries do not eliminate metastability; we assume indefinite retries for simplicity.} and switching to exponential backoff for retries eliminates the failure.
Revisiting the server’s passive scheduler is even better: prioritizing original requests over retries ({\em e.g.}, by maintaining separate queues) can eliminate metastability without the latency penalty of backoff.

\textbf{(3) Narrow Coverage.} It lacks coverage for instances of metastability that do not manifest as high latency and low goodput.
We observed several such incidents;
for example, metastability can manifest as oscillations in the number of active servers in a distributed system (\S\ref{sec:MFT}).
These diverse manifestations motivated us to move beyond the current characterization and develop the one presented in \S\ref{sec:faults}-\S\ref{sec:failures}.

% \subsection{The Oscillating Membership}\label{sec:incident-oscillating-membership}
% salam

% \textbf{Implementation.} salam

% \textbf{Demonstration.} salam

% \textbf{Malady.} salam

%% file: model.tex
\section{Model Sketch}\label{sec:model}
To reason about how components destabilize each other, we need a model that captures state changes, action effects, and how components read and write shared state.
We opt for intuition, and defer a formal model to Section~\ref{sec:formal-model} of the supplementary materials.

A system~$S=(\Sigma,\mathcal{A})$ is a state machine with a set of states~$\Sigma$, actions~$\mathcal{A}$, and variables.
States are valuations of variables.
~$S$ interacts with an \emph{environment} that, by taking actions of its own, can modify the values of some subset of the system's variables.
Environment actions are arbitrary---the system does not control the environment's behavior.
We denote state transitions of the system as~$s\rightarrow s'$, where~$s$ and~$s'$ are system states.

% A system~$S = (\Sigma_S, \mathcal{A}_S, \mathcal{V}, \mathcal{V}_E)$ is a state machine with a state space~$\Sigma_S$, a set of actions~$\mathcal{A}_S$, a set of variables~$\mathcal{V}$, and a set of variables~$\mathcal{V}_E$ such that~$\mathcal{V}_E\subseteq \mathcal{V}$---these are variables that the environment in which the system operates can modify.
% Each state is an assignment of values to the variables in~$\mathcal{V}$.
% Each state transition~$s\rightarrow^\alpha s'$ of the system, where~$s, s'\in\Sigma_S$, is accompanied by an action~$\alpha\in\mathcal{A}_S$.
% We omit the subscript~$S$ when the system is clear from the context.

A predicate~$P$ is a subset of the system's states, and an environment predicate~$E$ is a predicate that involves only the variables that the environment can touch.
Every execution of the system is a trace of states connected via state transitions, where at each transition at least one of the system and the environment take actions.
For every execution~$\sigma$ of the system, all suffixes of~$\sigma$ are also executions of the system.
%Whenever the environment has taken an action, followed by a system action, we say that the execution has taken a \emph{step}.
Given some environment predicate~$E$ and system states~$s$ and~$s'$, a transition~$s\rightarrow s'$ is an~$E$-step if~$s$ and~$s'$ satisfy~$E$.
If during a transition~$s\rightarrow s'$ only the system takes an action, say~$\alpha$, we further label the transition as~$s\rightarrow^\alpha s'$.

% Each system~$S = (\Sigma, \mathcal{A}, \mathcal{V}, \mathcal{V}_E)$ interacts with an environment that can modify the variables in~$\mathcal{V}_E$ by taking environment actions.
% Environment actions are arbitrary---the system does not control the environment's behavior.
% A predicate~$P$ is a subset of~$\Sigma$; an environment predicate~$Q$ is a predicate that involves only variables in~$\mathcal{V}_E$. 
% The system and the environment alternate actions: any execution~$\sigma$ of the system is a trace~$\sigma=s_0\rightarrow^{e_0} s_1\rightarrow^{\alpha_0} s_2\rightarrow^{e_1} s_3\rightarrow^{\alpha_1} s_4\rightarrow^{e_2}\dots$, where~$s_i\in\Sigma$ for all~$i\geq 0$,~$\{e_i\}_{i\geq 0}$ are environment actions, and~$\alpha_i\in\mathcal{A}$ for all~$i\geq 0$.
% Whenever the environment has taken an action, followed by a system action, we say that the execution has taken a \emph{step}.
% To keep notation short, we will omit mention of variables and write~$S=(\Sigma, \mathcal{A})$; it is to be inferred that the context implicitly indicates which variables the environment can modify.

If a system assigns a value to a variable following an action, we say the system \emph{writes to} the variable via the action, establishing a writes-to relation between them; if a system reads the value of a state variable with an action, we say that the system \emph{reads from} the variable via the action.

% If a system assigns a value to a variable following an action, we say the system \emph{writes to} the variable via the action, establishing a writes-to relation between them; if a system reads the value of a state variable with an action, we say that the system \emph{reads from} the variable via the action.

Given two systems~$S_1$ and~$S_2$, their composition~$S_1\parallel S_2$ is a system with the Cartesian product of~$S_1$ and~$S_2$'s states as its state space.
As for its actions, the composition takes an action when at least one of~$S_1$ and~$S_2$ take an action.
If, upon composition,~$S_1$ becomes responsible for writing to some variable of~$S_2$ that was previously written to by the environment, the environment stops writing to that variable in the composition.
We generalize this definition to compositions of more than two systems, and denote the composition of the~$n$ systems~$\{S_i\}_{0\leq i < n}$ with~$\parallel_{0\leq i < n}S_i$.
Whenever the composition takes an action during which each component~$S_i$ takes some action~$\alpha_i$ and the environment takes some action~$e$, the transition is serializable~\cite{papadimitriou1979serializability, bernstein1979serializability},~$i.e.$, the resulting state is equivalent to that resulting after \emph{some} serial execution of the actions~$\{\alpha_i\}_{0\leq i < n}$ and~$e$.

% Let~$\bot$ denote the absence of an action for any system.
% Given two systems~$S_1=(\Sigma_1, \mathcal{A}_1, \mathcal{V}_1, \mathcal{V}^1_E)$ and~$S_2=(\Sigma_2, \mathcal{A}_2, \mathcal{V}_2, \mathcal{V}^2_E)$, we define their composition as the system~$S_1\parallel S_2 = (\Sigma_1\times\Sigma_2, \mathcal{A}, \mathcal{V}_1\cup\mathcal{V}_2, \mathcal{V}^1_E\cup\mathcal{V}^2_E)$, where~$\times$ denotes the Cartesian product.
% As for~$\mathcal{A}$, assuming~$\alpha\in\mathcal{A}_1$ and~$\beta\in\mathcal{A}_2$, each action~$\gamma\in\mathcal{A}$ is of one of  ($i$)~$(\alpha, \beta)$, ($ii$)~$(\alpha, \bot)$, and ($iii$)~$(\bot, \beta)$: the composition makes a state transition whenever at least one component takes an action.
% We generalize this definition to compositions of more than two systems, and denote the composition of the~$n$ systems~$\{S_i\}_{0\leq i < n}$ with~$\parallel_{0\leq i < n}S_i$.
% Whenever the composition~$S$ takes an action~$\alpha=(\alpha_0,\dots, \alpha_{n-1})$, the action is serializable~\cite{papadimitriou1979serializability, bernstein1979serializability},~$i.e.$, the resulting state is equivalent to that resulting after \emph{some} serial execution of the actions~$\{\alpha_i\}_{0\leq i < n}$.
% In every execution of the composition~$S$, each system takes actions infinitely often.

Given a composition of~$n$ systems~$\{S_i\}_{0\leq i < n}$, we lift the writes-to relation between systems and variables to a writes-to relation between systems.
If a system~$S_i$, via some action~$\alpha_i$, writes to a state variable of~$S_i$ that system~$S_j$ reads from, or directly to a state variable of~$S_j$ that~$S_j$ reads from, we say that~$S_i$ writes to~$S_j$ via~$\alpha_i$.
This relation induces a \emph{composition blueprint}: a directed graph whose vertices are systems and whose edges represent writes-to relations between them.
If~$S_i$ writes to~$S_j$, the edge is from~$S_i$ to~$S_j$.
For a system~$S_i$, we call the set of all systems that write to it its \emph{writing neighbors}, and denote it with~$W_i$.
Similarly, we call the set of all systems that~$S_i$ writes to as its \emph{reading neighbors}, and denote it with~$R_i$.

To facilitate our proofs, which entail liveness assertions~\cite{alpern-schneider-liveness}, we assume that in every execution of a composition~$S$ ($i$) each component takes actions infinitely often, and ($ii$) during each transition at least one writing neighbor of every component takes an action.

% Given a composition of~$n$ systems~$\{S_i\}_{0\leq i < n}$, we lift the writes-to relation between systems and variables to a writes-to relation between systems.
% If a system~$S_i$, via some action~$\alpha_i\in\mathcal{A}_i$, writes to a state variable of~$S_i$ that system~$S_j$ reads from, or directly to a state variable of~$S_j$ that~$S_j$ reads from, we say that~$S_i$ writes to~$S_j$ via~$\alpha_i$.
% This relation induces a \emph{composition blueprint}: a directed graph whose vertices are systems and whose edges represent writes-to relations between them.
% If~$S_i$ writes to~$S_j$, the edge is from~$S_i$ to~$S_j$.
% For a system~$S_i$, we call the set of all systems that write to it its \emph{writing neighbors}, and denote it with~$W_i$.
% Similarly, we call the set of all systems that~$S_i$ writes to as its \emph{reading neighbors}, and denote it with~$R_i$.

%% file: faults.tex
\section{Metastable Faults}\label{sec:faults}

Dissecting the retry storm (\S\ref{sec:incident-retry-storm}) reveals two distinct contributors: the composition structure and the scheduler.
The former enables components to destabilize each other, while the latter repeatedly schedules for execution the destabilizing interactions.
We observe this separation in every incident we examined, which motivates us to characterize metastability by distinguishing  metastable {\em  faults} from metastable {\em failures}.
We focus on faults in this section; we defer the discussion of failures to \S\ref{sec:failures}.

Faults are structural vulnerabilities arising from compositions and capture the propensity for \emph{mutual destabilization} among \emph{individually stabilizing} components.
The composition of two components harbors a metastable fault if, in reaction to a \shock{}, the locally stabilizing actions of each component destabilizes the other.
Understanding metastable faults hinges on two questions: What are stabilizing systems? And are all faulty compositions of these systems about metastability? To answer these questions, we present ($i$) a notion of stabilization (\S\ref{sec:stabilizing-systems}), and ($ii$) a notion of compatibility (\S\ref{sec:compatibility}) that rules out trivially faulty compositions of stabilizing systems. We then define metastable faults (\S\ref{sec:metastable-faults}).
We illustrate these notions using the retry storm incident (\S\ref{sec:incident-retry-storm}) as an informal yet intuitive running example.
%We favor intuition, and defer formal versions of our definitions---both here and in~\S\ref{sec:failures}---to Section~\ref{sec:faults-failures-formal} of supplementary materials.  

\subsection{Stabilizing Systems}\label{sec:stabilizing-systems}

Consider a cluster manager tasked with maintaining a fixed number of active servers under a crash failure model.
A sufficiently strong \shock{}, {\em i.e.}, several workers crashing, can drive the system into a state where the cluster manager’s guarantee no longer holds.
Correctness requires the cluster manager to eventually reach a set of states with enough active workers and to stabilize there, regardless of the \shock's severity.

Such systems are known as self-stabilizing systems~\cite{dijkstra-1974-selfstabilizing, self-stabilization-dolev2000}.
Starting from an arbitrary state after a shock, a self-stabilizing system converges to a set of good states and remains there provided that there are no subsequent \shock{}s.
Inspired by this notion, we introduce the abstraction of a {\em potential function} to formalize metastable faults.
A potential function maps a system’s state to a value that quantifies its distance from the good states.
A system aiming to converge should take stabilizing actions that drive this potential to zero.
Whether every self-stabilizing system can be expressed using potential functions remains an open question for future work.

\begin{definition}[\textbf{Potential function}]\label{def:potential-function}
    For a system~$S$ with state space~$\Sigma$ and a state predicate~$G$ representing a set of good states, a function~$f:\Sigma\rightarrow\mathbb{R}_{\geq 0}$ is a \emph{potential function} for~$(S, G)$ iff:
    \begin{enumerate}[label=P\arabic*, ref=P\arabic*, left=0pt]
        \item\label{P1} for all~$s\in\Sigma$,~$f(s) = 0\Leftrightarrow s\in G$; and

        \item\label{P2} for all~$s, s'\in\Sigma$ and~$\alpha\in\mathcal{A}$, if the system makes a transition~$s\rightarrow^\alpha s'$, then~$f(s')\leq f(s)$.
    \end{enumerate}
\end{definition}

In other words, the potential function assigns zero to good states and positive values elsewhere, and every system action can only maintain or reduce this potential. The sole source of destabilization is the environment.
If the environment's destabilizing influence exceeds the system’s capacity to compensate, stabilization fails. Therefore, correctness for stabilizing systems follows an assume-guarantee form~\cite{assume-guarantee-jones1983}: starting from any state, the system will stabilize {\em as long as  the environment is well-behaved}.

\begin{definition}[\textbf{Stabilizing system}]\label{def:stabilizing-system}
    For a system~$S$, a predicate~$G$, and a potential function~$f$ for the pair~$(S, G)$, the pair~$(S, f)$ is \emph{stabilizing} iff there exists an environment predicate~$E$ such that, as long as the environment takes actions such that the system state repeatedly satisfies~$E$, eventually the system state also satisfies~$f = 0$, and keeps satisfying~$f=0$ as long as environment actions keep maintaining~$E$.
\end{definition}

A stabilizing system need not satisfy~$f=0$ initially. It suffices that, if the environment behaves well for long enough, there exist a time after which~$f=0$ holds as long as the environment keeps behaving well.
%the system has no way of causally predicting an environment misbehavior in the future, and therefore continues to satisfy~$f=0$ for one more execution step.
Inspired by the~$\rightarrow^+$ temporal modality of Abadi and Lamport~\cite{Abadi-Lamport}, we use the symbol~$\rightsquigarrow^+$ to denote our modality of choice: a pair~$(S, f)$ is stabilizing iff there exists an~$E$ such that~$E\rightsquigarrow^+ f = 0$.
We present a formal semantics for~$\rightsquigarrow^+$ in Section~\ref{sec:squiggly-semantics} of the supplementary materials.
% We defer a rigorous semantics for the temporal operator~$\rightsquigarrow^+$ to Appendix~\ref{sec:appendix-formal}.

\textbf{Running Example.} Our retry storm implementation (\S\ref{sec:incident-retry-storm}) naturally decomposes into two systems: the server,~$S_1 = (\Sigma_1, \mathcal{A}_1)$, and the retrier,~$S_2 = (\Sigma_2, \mathcal{A}_2)$. 
%We need to model them in a way that reveals their stabilizing nature.
For~$S_1$, a natural $\Sigma_1$ is the state space of the server’s queue, {\em i.e.}, each state is an ordering of requests and retries.
We pick one execution of the task \texttt{serve} as the only action for the server:~$\mathcal{A}_1 = \{\texttt{serve}\}$.
This action processes~$\min\{s,\ Q\}$ items per step ($s$ is service capacity,~$Q$ the queue size) and sends acknowledgments.
The environment injects requests and retries. Since the server repeats \texttt{serve} to drain load, a natural potential is~$f_1 = \max\{0, Q - s\}$,~$i.e.$, the overload on the server.
Choosing~$f_1 = Q$ would prevent stabilization at zero, because new requests always arrive.
$f_1$ is a potential function, as \texttt{serve} never increases it; only the environment's actions do so.%\af{Explain how this is connected to the code.}

For $S_2$, a natural $\Sigma_2$ consists of pending requests and their timers.
It aims to converge to a state where the only client requests pending are the newest.
We thus define $f_2 = \max\{0,\ P - P_{\text{threshold}}\}$, where $P$ is the number of pending requests and $P_{\text{threshold}}$ is a cutoff.~\footnote{Our analysis shows that~$P_{\text{threshold}} = (s - r)\cdot T$ is a suitable pick, with $r$ the nominal client load. Details are omitted for brevity.}

As for the retrier’s actions, we split the task \texttt{manage} into two parts, one adding requests to \texttt{pending} and one removing requests from \texttt{pending}.
We group the latter and the task \texttt{retry} into one action:~$\mathcal{A}_2 = \{\texttt{remove-n-retry}\}$.
This action never increases $f_2$: it removes acknowledged requests, updates timers for the ones in~\texttt{pending}, and retries timed-out ones.
Besides submitting acknowledgments, the environment takes care of adding new requests to \texttt{pending}, thus increasing~$f_2$.
 
%A natural pick for~$\Sigma_2$ is the set of all pending requests along with their corresponding timer values.
%It is trying to converge to a state where only the newest client requests are pending.
%It is thus intuitive to have~$f_2 = \max\{0, P - P_{threshold}\}$---$P$ is the number of requests in \texttt{pending} and~$P_{threshold}$ is some cutoff threshold.\footnote{We show in Appendix~\ref{sec:appendix-retry-storm-equations} that~$P_{threshold} = (s-r)\cdot T$ is a suitable pick, where~$r$ is the nominal client load.}
%Since no action of the retrier is allowed to increase~$f_2$, we pick~$\mathcal{A}_2 = \{\texttt{remove-n-retry}\}$.
%This action removes from~$\texttt{pending}$ all requests for which the retrier has received an acknowledgment in this step, updates the timers for all requests in~$\texttt{pending}$, and sends retries for the ones that timeout.
%The environment for the retrier in turn takes care of giving it client requests and adding them to \texttt{pending} ($i.e.$, increasing~$f_2$), and also giving it acknowledgments.
% Note how our model of the retrier delegates some of the functionality in its software (Figure~\ref{listing:retry}) to the environment.

To show that the server and the retrier are stabilizing, we have to determine environment predicates~$E_1$ and~$E_2$ such that~$E_1\rightsquigarrow^+ f_1 = 0$ and~$E_2\rightsquigarrow^+ f_2 = 0$,~$i.e.$, environments under which server and retrier stabilize.
If server's environment~$E_1$ always supplies fewer than~$s$ requests and retries combined per step, the service rate exceeds the input rate and the server will eventually reach~$f_1 = 0$.
For the retrier, if~$E_2$ provides more acknowledgments than new client requests per step, the retrier will remove pending requests faster than they accumulate, driving~$f_2$ to zero.

\subsection{Compatibility}\label{sec:compatibility}

When composing stabilizing systems, it must be {\em possible} for them to stabilize together.
We capture this sanity check with the notion of \emph{compatibility}: a set of stabilizing systems are compatible if, assuming that each system's writing neighbors have stabilized, the system itself can also stabilize under a well-behaved environment.
\begin{definition}[\textbf{Compatibility}]\label{def:compatibility}
    The stabilizing systems $(S_1, f_1)$,~$(S_2, f_2)$,~$\dots$, and~$(S_n, f_n)$ are \emph{compatible} if there exists an environment predicate~$E$ such that the following holds for~$\parallel_{0\leq i < n} S_i$, for all~$0\leq j < n$:
    \begin{enumerate}[label=C\arabic*, ref=C\arabic*, left=0pt]
        \item\label{C1}~$E\land (\bigwedge_{i\in W_j} f_i = 0)\rightsquigarrow^+ f_j = 0$. 
    
        % \item\label{C1}~$E\land f_1 = 0\rightsquigarrow^+ f_2 = 0$; and

        % \item\label{C2}~$E\land f_2 = 0\rightsquigarrow^+ f_1 = 0$.
    \end{enumerate}
    We call~$E$ a \emph{compatible environment} for the composition.
\end{definition}

Compatibility rules out trivially faulty compositions.
If two systems cannot assist each other in stabilizing, they are not meant to be composed.
We therefore define metastable faults only for compositions of compatible systems.

\textbf{Running Example.} Consider our composition of the server with the retrier: the retrier forwards both client requests and retries to the server, while the server sends acknowledgments back to the retrier.
The environment supplies the retrier with client requests; suppose it stops providing new requests.
Then the server eventually drains its queue and sends the correponding acknowledgments, leaving the retrier with no pending requests. 
If we denote this environment~$E$, we have~$E\rightsquigarrow^+ f_1 = 0$ and~$E\rightsquigarrow^+ f_2 = 0$, which implies~$E\land f_2 = 0\rightsquigarrow^+ f_1 = 0$ and~$E\land f_1 = 0\rightsquigarrow^+ f_2 = 0$.
Since the server's and the retrier's writing neighbors are~$W_1 = \{2\}$ and~$W_2 = \{1\}$, respectively, the retrier and the server are compatible; they \emph{can} stabilize together in the presence of a well-behaved environment.

\subsection{Enter Faults}\label{sec:metastable-faults}

Compatibility alone does not guarantee joint stabilization.
It suffers from a bootstrap problem: following a \shock{}, all systems may start in arbitrary states, with none stabilized.
Because stabilization in each system depends on their neighbors effecting a well-behaved environment, the composition may fail to stabilize if no component initially behaves as needed.
Indeed, a system $S_i$ might destabilize its reading neighbors while waiting for its writing neighbors to stop destabilizing {\em it}.
When such dependencies form cycles, the system becomes vulnerable to mutual destabilization, even as each component attempts to stabilize.
Metastable faults capture this emergent cyclic vulnerability.
In order to define a metastable fault, we first need to define a destabilizing action.

\begin{definition}[\textbf{Destabilizing action}]
    Let~$(S_1, f_1)$ and~$(S_2, f_2)$ be two stabilizing systems in the composition~$S$ of some compatible stabilizing systems, where~$\Sigma_2$ is~$S_2$'s state space.
    Let~$\alpha_1$ be an action with which~$S_1$ writes to~$S_2$, and let~$s, s'\in\Sigma_2$.
    The action~$\alpha_1$ is destabilizing at some state~$s\in\Sigma_2$ iff there exists a compatible environment~$E$ for~$S$ and an~$E$-step~$s\rightarrow^\alpha s'$ such that either~$f_2(s') \geq f_2(s) > 0$ or~$f_2(s') > f_2(s) = 0$.
    
    % Let~$(S_1, f_1)$ and~$(S_2, f_2)$ be two stabilizing systems in the composition~$S$ of some compatible stabilizing systems, where~$\Sigma_2$ is~$S_2$'s state space.
    % Let~$\alpha_1$ be an action with which~$S_1$ writes to~$S_2$, and let~$s, s', s''\in\Sigma_2$.
    % The action~$\alpha_1$ is destabilizing at some state~$s\in\Sigma_2$ iff there exists a compatible environment~$E$ for~$S$ and an~$E$-step~$s\rightarrow^e s'\rightarrow^{\alpha_1} s''$ such that either~$f_2(s'') \geq f_2(s) > 0$ or~$f_2(s'') > f_2(s) = 0$.
\end{definition}

Informally, an action is destabilizing iff it can \emph{possibly} increase or maintain potential somewhere in the system, even in the presence of a compatible environment.

\begin{definition}[\textbf{Metastable fault}]\label{def:metastable-fault}
    Given compatible stabilizing systems~$(S_0, f_0)$,~$\dots$,~$(S_{n-1}, f_{n-1})$, their composition~$S = \parallel_{0\leq i < n}S_i$ \emph{has a metastable fault} iff there exists a cycle of systems in the composition blueprint, such that for all systems in the cycle:
    \begin{enumerate}[label=M\arabic*, ref=M\arabic*, left=0pt]
        \item\label{M1} (\textbf{Writes-to}) it writes to the next system in the cycle via some action; and

        \item\label{M2} (\textbf{Destabilization}) said action is destabilizing at some state of the next system in the cycle.
    \end{enumerate}
\end{definition}

% \begin{definition}[\textbf{Metastable fault}]\label{def:metastable-fault}
%     Let~$i\oplus 1$ denote $i + 1\ \texttt{mod}\ k$ for~$0\leq i\leq k-1$ and some~$k$.
%     Given compatible stabilizing systems~$(S_0, f_0)$,~$\dots$,~$(S_{n-1}, f_{n-1})$, where $S_i = (\Sigma_i, \mathcal{A}_i)$, their composition~$S = \bigwedge_{0\leq i < n}S_i$ has a \emph{metastable fault} if there exists a subset~$M = \{S_{i_0}, \dots, S_{i_{k-1}}\}$ of the systems, and an action~$\alpha_{i_j}\in\mathcal{A}_{i_j}$ for~$0\leq j\leq k-1$, such that for all~$0\leq j < k$:
%     \begin{enumerate}[label=M\arabic*, ref=M\arabic*, left=0pt]
%         \item\label{M1} (\textbf{Cyclicity}) each~$S_{i_j}$ writes to~$S_{i_{j\oplus 1}}$ via~$\alpha_{i_j}$; and

%         \item\label{M2} (\textbf{Destabilization}) there exists a transition of~$S$, where only~$S_{i_j}$ takes an action, which is $\alpha_{i_j}$, changing~$S_{i_{j\oplus 1}}$'s state from~$s$ to~$s'$ for~$s, s'\in\Sigma_{i_{j\oplus 1}}$, such that~$f_{i_{j\oplus 1}}(s')\geq f_{i_{j\oplus 1}}(s)$.
        
%         % \item\label{M2} (\textbf{Destabilization}) there exists a transition of~$S$, where~$S_{i_j}$ takes action~$\alpha_{i_j}$ and~$S_{i_{j\oplus 1}}$ makes a transition from~$s$ to~$s'$ for~$s, s'\in\Sigma_{i_{j\oplus 1}}$, such that~$f_{i_{j\oplus 1}}(s')\geq f_{i_{j\oplus 1}}(s)$.
%     \end{enumerate}
% \end{definition}
Thus, the composition of compatible stabilizing systems has a metastable fault iff ($i$) there exists a \emph{writes-to cycle} among them, and ($ii$) each system in the cycle has an action that,  when taken in the presence of some compatible environment, either increases the next system's potential when it is zero, or does not decrease it when it is positive.

Locating such faults does not require a global analysis; it replaces reasoning about the entire composition with local, decompositional reasoning about pairs of components.
It suffices, for each component, to identify actions that fail to stabilize another given a compatible environment.
This involves, for all compatible environments, executing each action from an arbitrary starting state and verifying whether it is destabilizing.
A model checker can automate these pairwise checks, after which the designer can determine whether a cycle exists.

% Locating such faults does not require a causal analysis. It suffices, for each component, to identify actions that fail to stabilize another. This involves executing each action from an arbitrary starting state and verifying whether the reader's potential decreases. A model checker can automate these pairwise checks, after which the designer can determine whether a cycle exists.

\textbf{Running Example.} The retrier–server composition contains a writes-to cycle.
The retrier's action~$\texttt{remove-n-retry}$ never decreases the server's potential~$f_1$, and may even increase it. Similarly, the server's action~$\texttt{serve}$ can fail to decrease the retrier's potential~$f_2$ when it sends only useless acknowledgments---{\em e.g.}, when serving retries instead of new requests.
Thus, the composition exhibits a metastable fault.

%% file: failures.tex
\section{Metastable Failures}\label{sec:failures}

% Failures happen when the schedulers fan the fault into a self-sustaining loop where the scheduler prioritizes the destabilizing interactions, which in turn invoke the scheduler to make more bad choices.

A metastable failure occurs when a composition of compatible stabilizing systems fails to stabilize under a compatible environment.
In other words, some component experiences positive potential infinitely often despite a well-behaved environment and despite its own stabilizing tendency to converge to zero potential.
A concise way to express this is via temporal logic~\cite{temporal-logic-pnueli1977, manna1995temporal}, where~$\Box$ and~$\Diamond$ denote the \texttt{always} and \texttt{eventually} modalities.
Combined as~$\Box\Diamond P$ for some predicate~$P$, they assert that~$P$ happens always eventually,~{\em i.e.}, infinitely often.

\begin{definition}[\textbf{Metastable failure}]
    The composition~$S$ of the compatible stabilizing systems~$(S_0, f_0)$,~$(S_1, f_1)$,~$\dots$, and~$(S_{n-1}, f_{n-1})$ has a \emph{metastable failure} iff there exists a compatible environment~$E$ such that the composition admits an execution~$\sigma$ satisfying~$\Box E\land \Box\Diamond\neg(\bigwedge_{0\leq i < n}f_i = 0)$,~$i.e.$, an execution wherein the environment is always compatible but the system experiences positive potential infinitely often.
    We call~$\sigma$ a \emph{metastable execution} of~$S$.
\end{definition}

Metastable faults do not {\em necessarily} imply metastable failures,~$i.e.$, faults are not sufficient for failures.
Even as faults show structural destabilizing tendencies between components, compatibility implies the presence also of stabilizing tendencies: whether a potential will remain positive infinitely often depends on the {\em ordering} of these tendencies during executions.
Queues, the OS, and timers are among different parts of a computer system that affect ordering of events; we refer collectively to all such parts of the system that effect such orderings as the \emph{scheduler}.

If the scheduler favors destabilizing actions, potential will stay positive, triggering more destabilizing events and bad scheduling choices in a self-reinforcing loop.
A metastable failure will eventually emerge from this cycle of destabilizing actions and poor scheduling.
Conversely, if stabilizing actions are prioritized for {\em long enough}, compatibility will help drive potential to zero.
Once one system stabilizes, by compatibility it will in turn help stabilize its reading neighbors, until every component's potential reaches zero.
At that point, postponed destabilizing actions cease to exist---as we see in the example below---as compatibility implies that components do not destabilize each other in stability; all components, therefore, stabilize at zero potential.

\textbf{Running Example.} The pseudo-code in Figure~\ref{listing:retry} shows the scheduler's role in driving the system toward a metastable failure.
The retrier issues retries every $T$ timesteps, while the server places both retries and requests in a single queue and serves them in arrival order.
If a \shock{} overloads the server for long enough, the retrier's scheduling policy will cause it to flood the server with retries of pending requests.
As retries come to dominate the queue, the server increasingly wastes its resources by serving only old retries for which it has already responded to the original request.
The retrier’s aggressive scheduler amplifies load; the server’s passive scheduler lets this amplification erode effective capacity.
That degraded capacity starves the retrier’s pending requests, triggering yet another wave of retries and deepening the overload.

% In fact, as we show in Appendix~\ref{sec:appendix-retry-storm-equations}, the number of retries it sends increases linearly with the number of requests.

% Figure~\ref{fig:retry-storm-metastable-failure} illustrates this destructive interplay between workload and capacity, which is a function both of the composition and the schedulers.
% It contains both the stabilizing and the destabilizing tendencies, where the stabilizing ones arise from the compatibility between the two systems.

% \begin{figure}[tb]
%     \centering
%     \includegraphics[width=0.8\linewidth]{figures/retry-storm-metastable-failure.pdf}
%     \caption{The metastable failure for the retry storm. The red and green plus signs indicate the possibilities of destabilization and stabilization, respectively. The red arrow shows the coupling between workload amplification and capacity degradation that arises from the schedulers. We have included the server's stabilizing tendency for itself---as an edge from the server to itself---to show all the tendencies at play.}
%     \label{fig:retry-storm-metastable-failure}
% \end{figure}

A metastable failure can be avoided by changing how components schedule their actions. For example, the server could maintain two queues---one for retries and one for requests---prioritizing requests until enough acknowledgments are sent to the retrier, then switching to retries.
This ensures the retrier sees more acknowledgments than new requests, eventually driving the potential to zero; at that point, retries cease to exist as all pending requests have received acknowledgments, and all incoming requests will be served before timing out.
Alternatively, the retrier could employ exponential backoff, increasing its timeout until it pauses long enough for the server to clear its backlog and send pending acknowledgments. Either remedy allows the components to act for each other as the stabilizing environment each needs.

%A metastable failure can be avoided by changing how the components schedule their actions.
%For instance, the server can have two queues instead of one; one queue for the retries and one for the requests.
%It can then prioritize the requests over the retries, and serve just enough requests to send more acknowledgments to the retrier than client requests; it can then switch to serving retries.
%Then the retrier will receive more acknowledgments than requests, leading eventually to zero potential.
%Another remedy is for the retrier to change its scheduler by using, for instance, exponential backoff.
%Eventually, the timeout value becomes large enough that the retrier will be silent for long enough for the server to work through its backlog, and thus send all pending acknowledgments to the retrier, reducing both of their potentials to zero.
%Picking any one of these schedulers enables the components to act as the helpful environment the other needs to stabilize.

\subsection{Faults Are Necessary For Failures}\label{sec:necessary-conditions}

Our analysis establishes a fundamental relationship between metastable faults and metastable failures: while metastable faults are not sufficient to induce metastable failures, they are a necessary condition.
Specifically, if a composition of compatible stabilizing systems exhibits a metastable failure, then the composition has a metastable fault.
Although it may seem intuitive that a metastable failure cannot arise without faults, this is not automatic from our definitions.
We state this observation as a theorem together with a proof sketch here, deferring a rigorous proof to Section~\ref{sec:proofs} of supplementary materials.

\begin{restatable}{theorem}{faultsNecessary}\label{thm:faults-necessary}
    Let~$(S_0, f_0)$, $(S_1, f_1)$,~$\dots$, and~$(S_{n-1}, f_{n-1})$ be compatible stabilizing systems, and~$S=\parallel_{0\leq i < n}S_i$ their composition.
    If~$S$ has a metastable failure, then it has a metastable fault.
\end{restatable}
\begin{proofSketch}
    We prove the contrapositive by induction on~$n$, where we inductively assume that the claim holds for the composition of any~$k < n$ compatible stabilizing systems.
    For the base case~$n=1$, the claim holds trivially.
    Now for~$n > 1$, let~$E$ be any compatible environment for the composition~$S$ of compatible stabilizing systems~$(S_0, f_0)$, $(S_1, f_1)$,~$\dots$, and~$(S_{n-1}, f_{n-1})$.
    If~$S$ has no metastable faults, then there should exist some component~$S_j$ such that no matter what actions its writing neighbors take in the presence of~$E$, they reduce~$S_j$'s potential.
    Let~$J$ be the index set of all such components.
    Since any joint action of components in~$W_j$ is serializable, and~$S_j$ does not increase its own potential, we conclude that~$f_j$ will eventually reach 0 and stay there,~$i.e.$,~$E\rightsquigarrow^+ f_j = 0$.
    This implies~$E\rightsquigarrow^+ \bigwedge_{j\in J}f_j = 0$.
    Now, consider the system with indices not in~$J$:~$E\land (\bigwedge_{j\in J} f_j = 0)$ is a compatible environment for the remaining composition.
    We thus have~$E\land (\bigwedge_{j\in J} f_j = 0)\rightsquigarrow^+ \bigwedge_{j\in[n]\backslash J} f_j = 0$.
    Putting this besides~$E\rightsquigarrow^+ \bigwedge_{j\in J} f_j = 0$, we deduce~$E\rightsquigarrow^+ \bigwedge_{0\leq j < n}f_j = 0$, {\em i.e.}, the composition does not infinitely often experience positive potential in the presence of a compatible environment.
    \qed
\end{proofSketch}

This result generalizes prior work showing that cycles in the composition blueprint---{\em i.e.}, feedback loops---are necessary for metastable failures~\cite{hotnets2025}. Our contribution builds on this insight by introducing an automatically verifiable criterion for excluding metastable failures: construct the composition blueprint and check for cycles of destabilizing actions. If no such cycles are present, then, under the assumed potential model, the system is not susceptible to metastable failures.

The proof above relies critically on the assumption that joint actions are serializable. Without this assumption, the potential of a component becomes ill-defined when neighboring components perform concurrent writes. The harmful effect of circular destabilization emerges only when actions have well-defined and isolated impacts on the potential of individual components.

%% file: macroscopics.tex
\section{Achieving Metastable Fault Tolerance}\label{sec:methodology}

The previous sections explain metastability using abstractions, yet real failures occur in production systems with tens of  thousands of lines of code. Bridging this gap is challenging and demands deep, cross-stack expertise, as some faults couple together components from across the stack.
Guided by the goal of analyzing metastability in production code, we take a modest first step by sketching a design methodology for building MFT systems, and apply it to a modeled version of a novel incident (\S\ref{sec:MFT}).
The methodology revolves around proving that a system is MFT: it ($i$) extracts from the system the object of the proof, called the \emph{metastability skeleton}, thereby also eliminating boilerplate code, ($ii$) introduces a way to pick and test candidate potential functions, ($iii$) augments the skeleton with suitable scheduling mechanism and policy, and ($iv$) outlines a proof blueprint.
%This methodology offers three benefits: ($i$) it simplifies modeling by eliminating boilerplate code; ($ii$) it enables augmenting a system with the right scheduling to achieve metastable fault tolerance without altering system structure; and ($iii$) it provides a generic rubric for formal reasoning about metastable fault tolerance.
Our methodology draws on experience studying diverse metastable failures, and we have developed the \dsl{} toolkit\footnote{Link removed for anonymity.} to support it.

\textbf{Step 1: Derive a metastability skeleton.} In practice, components destabilize one another by exchanging different types of requests.
Upon receiving a request, a component undergoes a sequence of state transitions---possibly altering its potential---and may in turn issue requests to other components.
The occurrence of a failure depends on three factors: ($i$) the scheduling of in-flight requests, ($ii$) the scheduling of internal state transitions, and ($iii$) the causal relationships among different request types. An operational model that captures these aspects enables us to isolate what matters for metastability while abstracting away irrelevant system details in the interest of simplicity and tractability.
We refer to this model as the \emph{metastability skeleton}.

% Metastable failures happen for two main reasons: components can destabilize each other, and the scheduler prioritizes the destabilizing actions  over the stabilizing ones.
% An operational model thus requires communication between the components, and also a way to schedule communication and internal state transitions.

Our operational model of choice uses queues and resources to implement scheduling; queues schedule requests upon their arrival at the destination, while resources control the scheduling of processes and threads. 
The skeleton also requires minimal logic for generating and linking different types of requests,~$e.g.$, upon receiving a request of type \texttt{client}, a server generates a request of type \texttt{ack} and sends it back.
Queue capacities, resource availability, and timers are especially important in this model, as they govern the flow of requests around the system.

The \dsl{} DSL (\S\ref{sec:nimbus-primer}) exposes these abstractions directly:
it supports message passing between agents, provides input and output queues for each agent, enables task scheduling based on resource availability, and offers common programming constructs ($e.g.$, while loops) to implement control logic.
By making scheduling of messages and tasks explicit, \dsl{} pools together parts of a computer system that are usually not present within the same codebase,~$e.g.$, it enables the designer to include in the \dsl{} model some details of OS scheduling that actual production code is oblivious to.
Figure~\ref{listing:retry} in~\S\ref{sec:incident-retry-storm} is an example \dsl{} syntax, demonstrating both application logic and task scheduling in an intuitive way.

\textbf{Step 2: Study dynamics.} Metastable failures are, after all, one among many dynamic behaviors of systems; thus, any method that exposes system dynamics should, in principle, reveal metastable failures—even without explicitly labeling them as such.
Therefore, a lack of detailed insight into faults does not preclude the ability of visualizing failures at design time, and thus alerting the designer of a potential fault.

One particularly effective method is the use of \emph{vector fields}~\cite{analyzing-metastable-failures}.
A vector field represents the system’s tendency at each state---relative to some state function---using a vector. For example, in systems modeled by differential equations, the vector indicates the direction and magnitude of the derivative at each point in the state space. By illustrating dynamic tendencies from arbitrary states, vector fields help analyze behavior after a \shock{} perturbs the system.

%A vector field depicts the system's tendency at each state---with respect to some function of the state---using a vector.
%For instance, if the system is expressed using differential equations, then the vector demonstrates the direction and the magnitude of the derivative at each point in the state space.
%By demonstrating dynamic tendencies starting from any state, vector fields make it easier to study a system's behavior after a \shock{} puts it in an arbitrary state.
%By demonstrating the dynamic tendencies starting from any state, vector fields encompass various \shock{}s \af{Change this.}.

It is neither necessary nor feasible to use the full system's state---with queue contents, local variables, and in-flight requests---as the vector field state space. Instead, our vector field represents how the system's {\em potential} evolves as the system interacts with a compatible environment post-\shock{}.  
This narrower focus  reduces dramatically the size of the space, enabling a 2D or 3D representation of the vector field.  Within this reduced space, a metastable failure manifests as two coexisting tendencies: one stabilizing at zero and another diverging toward positive potential. Designers can thus predict metastable failures by inspecting the vector field.

% The resulting behavior depends on the incident; in the retry storm, it manifests as a blow-up in potential, while in the case study in \S\ref{sec:MFT} it manifests as oscillation.

The challenge from a modeling perspective lies in identifying suitable candidates for potential. The potential function for a stabilizing component should derive from the component's specified goal: for example, a server aims to serve load, so excess queue load is a natural candidate; a load balancer seeks to distribute load evenly, so any measure of imbalance is appropriate. Definition~\ref{def:potential-function} further simplifies this task by requiring that a component’s actions do not increase its potential; if they do, the component should be divided into smaller units, each with an appropriate potential function.

Similar to Metafor~\cite{analyzing-metastable-failures}, the \dsl{} toolkit provides a \texttt{canvas} mode that enables visualizing a system's vector field---but with an important difference.
Metafor derives its vector fields by modeling systems as continuous-time Markov chains, implying that the system has no memory (in the statistical sense).
Our experience suggests instead that systems experiencing metastable failures have memory, and in fact that memory plays an important role, because it affects scheduling.
Thus, \dsl{} generates its vector field using the system's executions as its source of ground truth, alleviating the need for explicit mathematical modeling.
Designers have to annotate the vector field state variables in the \dsl{} code expressing their system.
Then, \dsl{} ($i$) executes multiple system runs starting with a random \shock{}, ($ii$) records for each state~$s$ the set~$\texttt{next}(s)$ of states reached immediately after~$s$ across all runs, and ($iii$) computes the system’s tendency at~$s$ as the average of all states in~$\texttt{next}(s)$, represented by a vector from~$s$ to this average.
\dsl{} uses lock-step execution semantics to capture dynamic tendencies, mirroring the assumption, implicit in differential equations, that variables update simultaneously based on current values.
Additionally, \dsl{} allows designers to specify the \shock{} using programming constructs,~$e.g.$, using the keyword \texttt{inject}, the designer can insert extra requests in a server's queue.

%Unlike existing work---Metafor~\cite{analyzing-metastable-failures}---\dsl{} avoids relying on a mathematical model of the system , and takes the system's executions as its source of ground truth.
%Metafor generates vector fields by first modeling a system using a continuous-time Markov chain, which implies that the system has no memory (in the statistical sense).
%While helpful, our experience shows that systems experiencing metastable failures have memory, and in fact memory plays an important role because it affects scheduling.

Figure~\ref{fig:retry-storm-vector-field} illustrates how this visualization can guide designers in identifying the failure for the retry storm example (\S\ref{sec:incident-retry-storm}); the x-axis represents the retrier’s potential (pending requests) and the y-axis represents the server’s potential (queue size).
Assuming a maximum queue capacity of 500 requests, the system exhibits, as expected, two coexisting yet conflicting tendencies: one toward zero and another diverging toward saturation and ever-increasing pending requests.

\begin{figure}[t]
    \centering
    \includegraphics[width=\columnwidth]{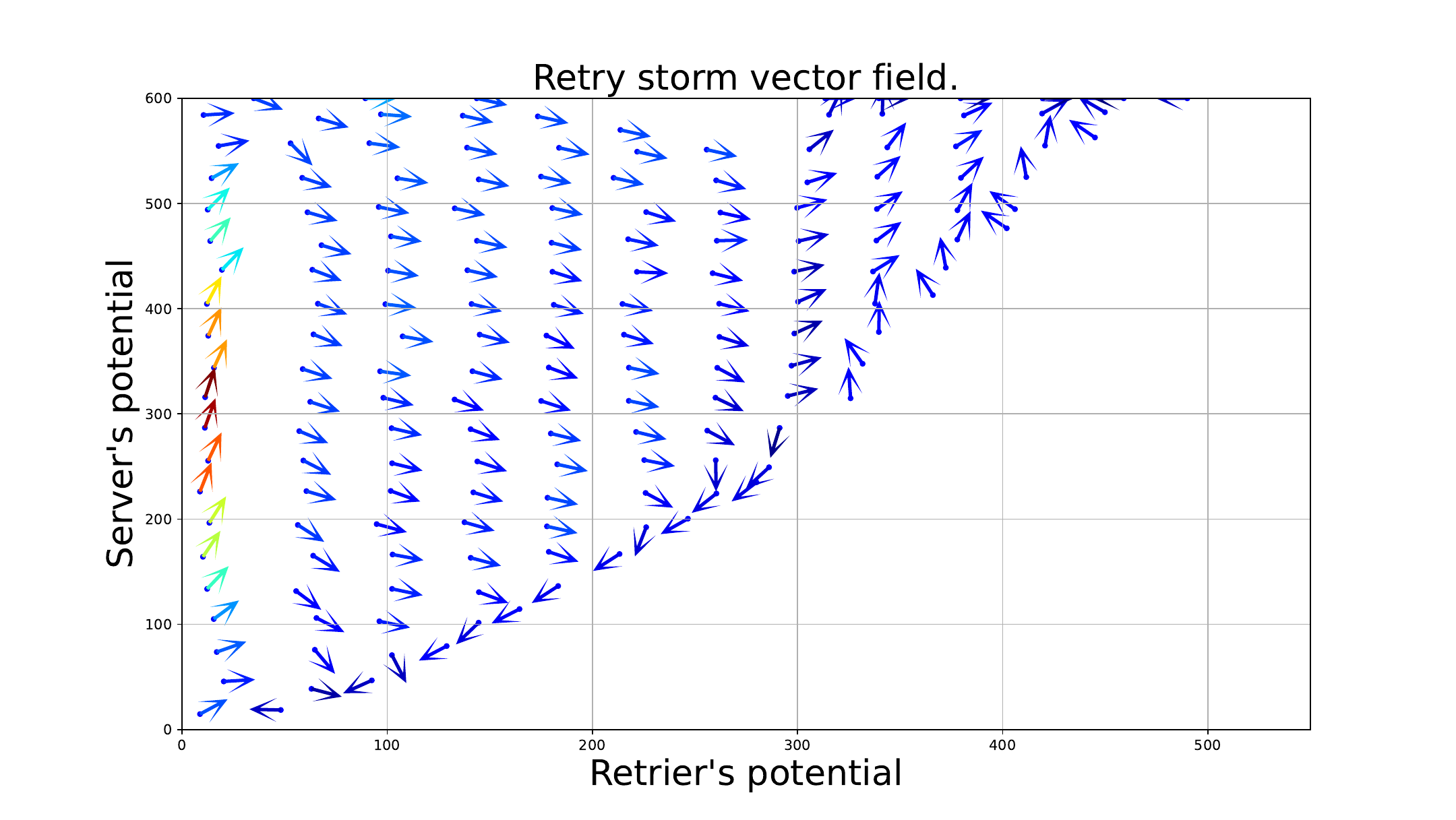}
    \caption{The \dsl{} vector field for the retry storm. Arrow color indicates tendency intensity at that point; warmer means higher intensity.}
    \label{fig:retry-storm-vector-field}
\end{figure}

% \begin{figure*}[t]
%     \centering
%     \includegraphics[width=\textwidth]{figures/retry-storm-join-vector-field.pdf}
%     % \begin{subfigure}[b]{0.45\textwidth}
%     %     \centering
%     %     \includegraphics[width=\columnwidth]{figures/retry-storm-vector-field.pdf} % Replace with your image
%     %     \rule{0.8\textwidth}{2cm} % Placeholder rule
%     %     %\caption{No backoff.}
%     %     %\label{fig:A}
%     % \end{subfigure}
%     % \hfill
%     % \begin{subfigure}[b]{0.45\textwidth}
%     %     \centering
%     %     \includegraphics[width=\columnwidth]{figures/retry-storm-exp-backoff-vector-field.pdf} % Replace with your image
%     %     \rule{0.8\textwidth}{2cm} % Placeholder rule
%     %     %\caption{Exponential backoff.}
%     %     %\label{fig:B}
%     % \end{subfigure}
%     %\caption{\dsl{} vector fields for the retry storm. The color of an arrow indicates the intensity of the tendency at that point.}
%     \label{fig:retry-storm-vector-field}
  
%     % \subfloat[No backoff.]{\includegraphics[width=0.45\columnwidth]{figures/retry-storm-vector-field.pdf}\label{fig:retry-storm-no-backoff--vector-field}}
%     % %\hspace{5mm}
%     % \subfloat[Exponential backoff.]{\includegraphics[width=0.47\columnwidth]{figures/retry-storm-exp-backoff-vector-field.pdf}\label{fig:retry-storm-exp-backoff-vector-field}}
%     % \vspace{-2mm}\caption{\dsl{} vector fields for the retry storm.}\label{fig:retry-storm-vector-field}
% \end{figure*}

\textbf{Step 3: Manage scheduling.} Our characterization of metastability identifies two approaches to achieving metastable fault tolerance: ($i$) eliminate the fault or ($ii$) avoid it through appropriate scheduling. While appealing, the first option is often impractical or even impossible without post-mortem insight.
For example, in one incident an air conditioning crash raised the temperature, causing (i) a worker’s CPU to overheat, prompting (ii) a thermal manager to put the worker to sleep, triggering (iii) a load balancer to misinterpret the worker as idle and assign to it  more load, leading to (iv) sustained overheating due to overload---repeating indefinitely~\cite{understanding-vicious-cycles-khan2011}. The fault stemmed from an interaction between load and {\em temperature}, a factor designers rarely model.

Short of modeling the entire environment, the most viable strategy is to manage scheduling carefully when composing stabilizing systems.
This requires explicitly declaring systems as stabilizing where appropriate, which introduces stabilizing and destabilizing interactions.
The system's scheduling must satisfy two properties: (\textbf{R1}) schedule destabilizing actions tentatively 
and (\textbf{R2}) prioritize stabilizing actions over destabilizing ones until the composition stabilizes.
R2 ensures convergence to globally zero potential, while R1---combined with compatibility---guarantees that once the system is globally stable, no further destabilizing actions occur; postponed tentative destabilizing actions disappear.
Without this guarantee, scheduling merely reorders increments and decrements, leaving the final potential unchanged.
Note that scheduling decisions of import span the entire stack: timers, OS scheduling, network scheduling, etc.
The main merit of the metastability skeleton---and \dsl{}---is striking a balance between generality and feasibility in collecting some of such parts of the system into one specified object of inquiry.

In the retry storm example, exponential backoff is one way to implement R1 and R2: the growing backoff postpones retries until they effectively cease, allowing the server to process pending acknowledgments and clear the retrier’s queue, thus altogether eliminating the need for further retries.

\textbf{Step 4: Complete the proof.} 
A system is metastable-fault-tolerant if it stabilizes despite having a metastable fault.
Proving metastable fault tolerance requires showing that, from any arbitrary state that a \shock{} may leave a composition of stabilizing systems in, the system eventually stabilizes globally.
The first step is to define the scope of the \shock{}: the most general model would allow a \shock{} to land the system in any arbitrary state, while more realistic models would constrain the state space to those reachable under a specific adversary.
The \shock{} has to be finite: there must exist an unknown time after which ($i$) adversarial interactions cease, and ($ii$) all remaining adversarial effects are healed ($e.g.$, crashed nodes are uncrashed).
This is akin to the partial synchrony assumption~\cite{dwork1988partialsynchrony} made for consensus protocols, where the network is eventually required to be synchronous forever.
Specifying an adversary and proving stabilization establishes, in the same breath, tolerance against all faults triggered by the adversary.

% The initial state after the \shock{} is not the only artefact of the \shock{}; there might be some events pending right after the \shock{} ends, events for which there would be no explanation if we simply look at the initial state.
% Therefore, the designer should assume an arbitrary initial state \emph{and} an arbitrary list of pending events.
% Both the state and the list satisfy certain invariants, depending on how strong the adversary is \af{Explain more, and add an example.}.

Using the scheduler properties outlined above, the designer must prove that ($i$) the system's potential never increases and ($ii$) sometimes decreases, and that ($iii$) no destabilizing events occur after global stabilization. Together, these guarantees imply metastable fault tolerance; techniques from the self-stabilization literature can be of great help here.
In~\S\ref{sec:MFT}, we apply this methodology to design a metastable-fault-tolerant cluster manager for a novel incident.

%% file: case-2.tex
\section{Case Study: Oscillating Membership}\label{sec:MFT}

We use an incident from a major commercial gaming platform with hundreds of millions of users to show our methodology in action.
Consider a cluster of workers~${w_0, w_1, \dots, w_n}$ managed by a cluster manager (\texttt{CM}).
Each worker is in one of three states: active, idle, or crashed.
Active workers serve client requests; idle workers do not; crashed workers remain crashed until the environment restores them to idleness.
Client load originates from an external load balancer (not modeled here) and is inversely proportional to the number of active workers~$W$: as~$W$ decreases, the load per active worker increases.

\texttt{CM} aims to maintain at least~$W_{min}$ active workers to tolerate crashes.
It listens for periodic heartbeats from workers it considers active and, upon a timeout, issues ($i$) a \texttt{sleep} command to the unresponsive worker and ($ii$) a \texttt{wakeup} command to an idle worker.
Workers transition to active upon receiving \texttt{wakeup} and to idle upon receiving \texttt{sleep}.

During the incident metastability manifested as persistent oscillation in~$W$.
After a \shock{} that crashed some workers, the cluster manager failed to restore~$W$ to~$W_{min}$; instead,~$W$ oscillated as workers were repeatedly put to sleep and awakened.
This behavior persisted even after crashed workers recovered, and packet loss and queueing delays did not play a role.
Nothing apparent in the worker or cluster manager implementations suggests the root cause. %\ah{An illustration with the actual membership numbers, to show the oscillation, would be nice, if we can fit it in.}

\textbf{The Culprits.} Two implicit aspects of the system interact to create a cycle leading to oscillation: ($i$) overloaded active workers miss heartbeats ({\em e.g.}, the service thread starves the heartbeat thread), and ($ii$) \texttt{wakeup} commands take longer to execute than \texttt{sleep} commands ({\em e.g.}, waking up involves provisioning virtual machines, whereas sleeping is a simple state change).
The failure unfolds as follows: after some workers crash, ($i$) remaining active workers become overloaded; ($ii$) \texttt{CM} detects crashes and issues \texttt{wakeup} commands; ($iii$) before new workers activate, overloaded workers miss heartbeats, causing \texttt{CM} to put them to sleep; ($iv$) new workers finally wake up, only to become overloaded because the old ones are now idle---this patterns repeats indefinitely.

\begin{figure}[t]
    \centering
    \includegraphics[width=\columnwidth]{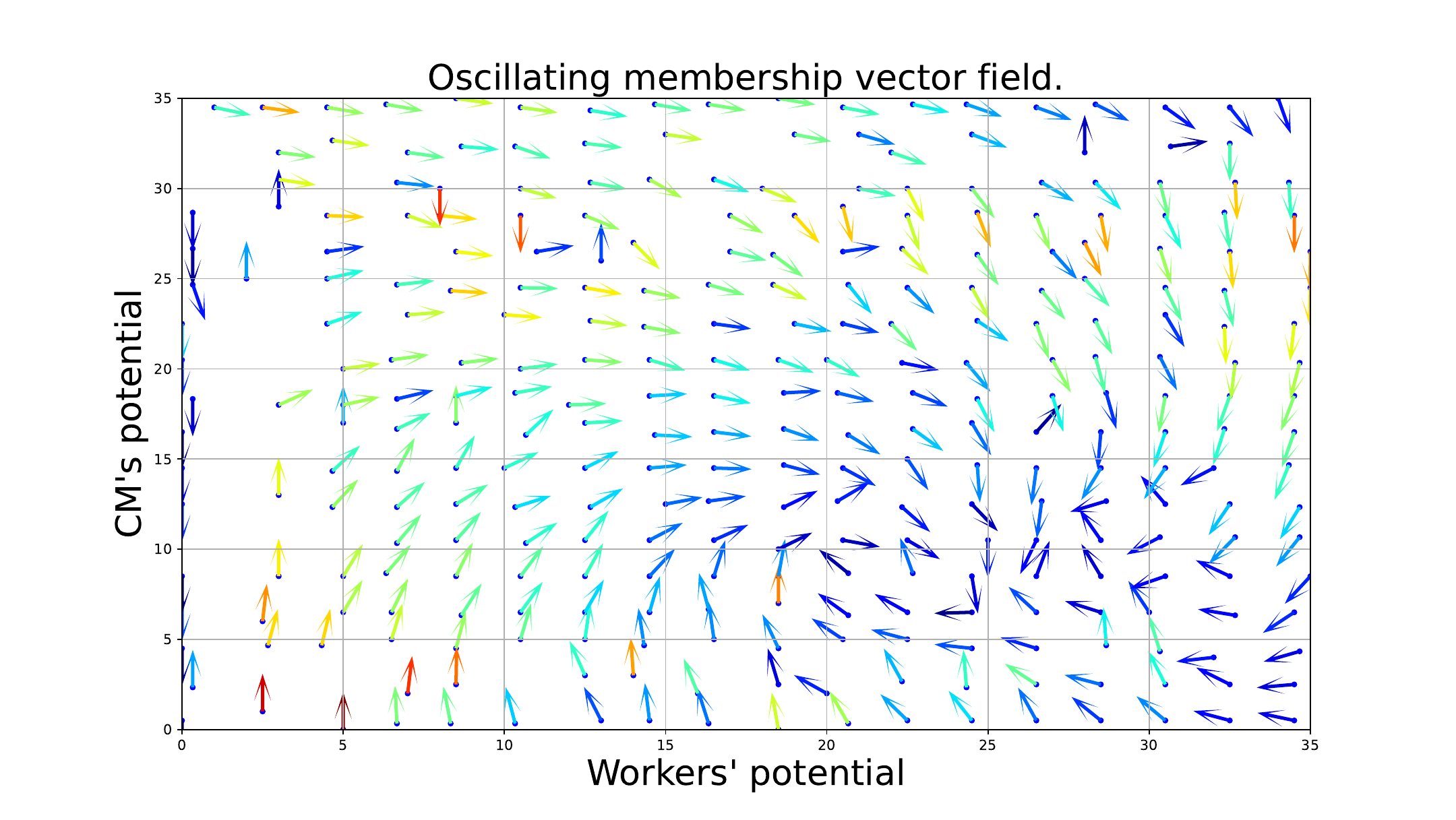}
    \caption{The \dsl{} vector field for the oscillating membership case. % Arrow color indicates tendency intensity at that point; warmer means higher intensity.
    }
    \label{fig:oscillating-membership-vector-field}
\end{figure}

\textbf{Potential Functions.} Both workers and \texttt{CM} respond to a \shock{} as stabilizing systems: workers aim to handle excess load, and \texttt{CM} seeks to maintain sufficient active workers.
While load is a natural candidate for worker potential, a useful proxy is~$f_1 = W_{noOverload} - W$, where~$W_{noOverload}$ is the minimum number of active workers required to avoid missed heartbeats.
Workers require an environment that activates enough peers to drive~$f_1$ to zero---here, \texttt{CM} and the crashing environment together form that environment.
For \texttt{CM}, let~$W'$ denote the number of workers it \emph{deems} active; then~$f_2 = W_{min} - W'$ is a suitable potential.\footnote{No action of \texttt{CM} increases~$f_2$, as \texttt{sleep} commands decrease~$W$ and not~$W'$.}
Notably, \texttt{CM} cannot reduce~$f_2$ by its own actions, and relies on timely heartbeats to stabilize.
The functions~$f_1$ and~$f_2$ are immediately useful for checking compatibility and exposing the fault.

\textbf{Compatibility.} Let~$E$ denote an environment that never crashes workers.
Under~$E$, if~$f_2 = 0$ holds, then \texttt{CM} believes enough workers are active, implying it receives sufficient heartbeats and thus issues no \texttt{sleep} commands.
Consequently,~$f_1 = 0$; formally,~$E \land f_2 = 0 \rightsquigarrow^+ f_1 = 0$.
Similarly, if~$f_1 = 0$ holds under~$E$, workers are never overloaded, and if~$f_2 > 0$, \texttt{CM} will activate additional workers until~$f_2 = 0$; thus,~$E \land f_1 = 0 \rightsquigarrow^+ f_2 = 0$.
This establishes compatibility.

% Note how these potential functions do not include implementation details.
% In fact,~$f_1$ is a function of the state of all workers combined, and is not an explicit state variable anywhere in the system.
% They are nonetheless useful abstractions, as they reveal the metastable fault.
% The workers increase~$f_2$ by taking an action representing missing a heartbeat---explicit in the model by virtue of expecting crashes---which leads to timeouts at \texttt{CM} and a decrease in~$W'$.
% \texttt{CM} increases~$f_1$ by sending \texttt{sleep} commands to workers that might be active, hence decreasing~$W$.

\textbf{Metastable Fault.} Workers destabilize \texttt{CM} by missing heartbeats, while \texttt{CM} destabilizes workers by issuing \texttt{sleep} commands.
Model \texttt{CM}’s actions as~$\texttt{wakeup}(i)$ and~$\texttt{sleep}(i)$: the former, if~$w_i$ is idle, decrements~$f_1$ (a stabilizing action), while the latter, if~$w_i$ is active, increments~$f_1$ (a destabilizing action).
A worker action~$\texttt{miss}(i)$ increments~$f_2$ by making \texttt{CM} believe~$i$ has crashed.
The possibility of a worker missing heartbeats---and thus destabilizing \texttt{CM}---is inherent in the functional composition of workers and \texttt{CM}; indeed, \texttt{CM} anticipates the possibility of missed heartbeats due to crashes.
Locating the fault therefore requires only a careful static analysis of the interface between the components.

\textbf{Metastable Failure.} The schedulers turn the fault into a failure: overloaded workers starve heartbeats, and \texttt{CM}, by issuing \texttt{sleep}(i) and \texttt{wakeup}(i) simultaneously, implicitly allows \texttt{sleep} to take effect before \texttt{wakeup}.
This time gap gives the scheduler of the remaining active workers, which are now overloaded, time to destabilize \texttt{CM} by in turn missing heartbeats, prompting further \texttt{sleep} commands---still within the time gap between \texttt{sleep} and \texttt{wakeup} commands.
By the time the first \texttt{wakeup} occurs, the system has drifted even farther from stability.

% \begin{figure}[tb]
%     \centering
%     \includegraphics[width=1.1\linewidth]{figures/oscillating-membership-failure.pdf}
%     \caption{The metastable failure for the oscillating membership incident. The symbols ``O.L.'', ``s'', and ``w'' indicate overload, sleep, and wakeup, respectively. The colors red and green signify destabilization and stabilization, respectively.}
%     \label{fig:oscillating-membership-metastable-failure}
% \end{figure}

\subsection{An MFT Cluster Manager}

We apply the method from \S\ref{sec:methodology} to this incident.

\textbf{Step 1: Derive metastability skeleton.} We implement a system similar to the one we received the report on using the \dsl{} DSL.
The implementation has two agents, one for \texttt{CM} (129 LoC) and one representing the entire ensemble of workers (116 LoC).
The workers' agent simulates the wakeup delay, and also omits heartbeats if~$W < W_{noOverload}$.
Bundling workers together makes it easier to detect overload by just counting their number and abstracting load away.
Designers can also choose to implement each worker as a separate agent, and to use internal scheduling to starve heartbeats; the results would not change.
The idiomatic way of writing \dsl{} code helps the designer specify the metastability skeleton.

\textbf{Step 2: Study dynamics.} Figure~\ref{fig:oscillating-membership-vector-field} shows the resulting vector field for a cluster of 40 workers with fixed timeouts.
\dsl{} generated this field from 1,000 executions, each simulating 100 seconds, in about three minutes.
The field reveals two opposing tendencies, one pulling toward zero potential and the other pushing away, resulting in a persistent oscillation.
This visualization makes the risk of oscillation immediately apparent and prompts deeper reasoning about the scheduler.

\textbf{Step 3: Manage scheduling.} We demonstrate that it is possible to build an MFT cluster manager for this incident \emph{without} knowing that workers starve heartbeats,~$i.e.$, by focusing on the cluster manager as a stabilizing system.
We make a simplifying assumption: overload is an instantaneous function of~$W$---the moment~$W_{min} - W = 0$, active workers are no longer overloaded.

The cluster manager (\texttt{CM}) issues two commands: \texttt{sleep}, which is destabilizing, and \texttt{wakeup}, which is stabilizing.
Our methodology (\S\ref{sec:methodology}) rests on two principles: destabilizing actions must be tentative (\textbf{R1}), and stabilizing actions must take precedence (\textbf{R2}). %\ah{I still wish that this were somehow crisper and easier to understand. It's fine here, because you're taking the reader by the hand and are guiding them towards the solution. But if I were on my own with a different system that I didn't fully understand, I think I would have a hard time... which potential function do I pick? What are the 'stabilizing' actions? How do I make the other ones 'tentative', and what exactly does that mean?}
To enforce these principles, we introduce \mftcm{}, a cluster manager that, by an appropriate combination of new mechanisms and policies, achieves metastable fault tolerance even when workers starve heartbeats.
Importantly, \mftcm's new scheduling policies and mechanisms  do not alter the cluster manager’s functionality: they only control the timing of actions. 

\noindent\underline{\em New scheduling mechanisms and policy} Because \mftcm{} cannot control how the environment orders \texttt{sleep} and \texttt{wakeup} events, the designer must estimate empirically the time needed by each command to take effect, denoted, respectively, as $d_\texttt{sleep}$ and $d_\texttt{wakeup}$.
For simplicity, we assume $d_\texttt{sleep}=0$.

To satisfy R1 and R2, \mftcm's scheduling mechanisms ensure that \texttt{sleep} events remain tentative, that they can be delayed until sufficient \texttt{wakeup} events occur, and that all \texttt{wakeup} events target idle workers.
Policy complements this by guaranteeing that delays for \texttt{sleep} events are long enough that they will follow \texttt{wakeup} events.
To achieve this, \mftcm{} introduces an artificial delay $d'_\texttt{sleep}$; without it, \texttt{sleep} would precede \texttt{wakeup} since $d_\texttt{wakeup}>d_\texttt{sleep}=0$.

%\vspace{-1em}
\noindent\underline{\em A Refined View of Workers' States} To implement these guarantees, \mftcm{} refines its view of worker states. Instead of modeling them as either \texttt{idle} or \texttt{active}, it uses five states: $\{\texttt{idle}, \texttt{active}, \texttt{waking-up}, \texttt{snoozing}, \texttt{pending-timeout}\}$.
When a worker $w_i$ times out, \mftcm{} checks for an idle worker $w_j$. If one exists, it issues $\texttt{wakeup}(j)$ and marks $w_j$ as \texttt{waking-up}. Upon confirmation, $w_j$ becomes \texttt{active}; otherwise, it reverts to \texttt{idle}. Simultaneously, $w_i$ moves to \texttt{snoozing}, but $\texttt{sleep}(i)$ is delayed by $d'_\texttt{sleep}$. If no idle worker exists, $w_i$ enters \texttt{pending-timeout}, and \mftcm{} retries until a heartbeat arrives or an idle worker is found.

The intermediate states \texttt{waking-up} and \texttt{snoozing} are essential to prevent issuing \texttt{wakeup} to non-idle workers.
%Consider a cluster with five workers, three of which must remain active.
If workers were modeled only as either \texttt{idle} or \texttt{active}, overlapping \texttt{sleep}/\texttt{wakeup} actions could cause a \texttt{wakeup} command to be sent to an overloaded worker, hence resulting in no decrease in potential.
Similarly, \texttt{pending-timeout} ensures eventual recovery to $W_{\min}$ active workers by enabling \mftcm{} to remember timed out workers and to retry until either idle workers are found or heartbeats resume.

\noindent\underline{\em Choosing a delay policy}  Avoiding metastable failures requires {\em all} \texttt{sleep} events to happen after any \texttt{wakeup} event until stability is restored post-\shock{}.
Thus, the value of~$d'_\texttt{sleep}$ is critical.
Setting~$d_\texttt{wakeup}<d'_\texttt{sleep}$ may seem sufficient, but it does not guarantee stabilization.
It only ensures the correct ordering for every pair of \texttt{wakeup} and \texttt{sleep} events issued in response to a timeout, and a poor choice of~$T$ can still lead to a \texttt{sleep} event happening before an unrelated \texttt{wakeup} event, and cause a metastable failure.
Instead, as we show in the following proof, we require~$d_{wakeup} + T < d'_{sleep}$. 

\noindent\underline{\em The Protocol} Protocol~\ref{protocol:cm} illustrates \mftcm{}.
The protocol proceeds in steps; it comprises three procedures driven by \textsc{Main}, which is executed in each step. \textsc{Uncrashed} runs asynchronously when a worker notifies \mftcm{} that it is idle after a crash; heartbeat handling is omitted for brevity.
%\la{Do you think these observations are essential?}
Two practical observations follow. First, communication between workers and \mftcm{} occurs only during state changes. Since active worker counts oscillated in the reported incident, we infer that this messaging infrastructure is unaffected by overload and thus reliable.
Second, although our implementation uses a single thread for timer updates, timers can run concurrently in practice. \mftcm{} itself is not under load and can schedule timers successfully.

{
\floatname{algorithm}{Protocol}
\begin{algorithm}[H]
    \caption{The MFT cluster manager~\mftcm{}.}%
    \label{protocol:cm}
    \begin{algorithmic}[1]\scriptsize
        \LeftComment{State variables:}
        \State $w \gets [\texttt{active}, \dots, \texttt{active}, \texttt{idle}, \dots, \texttt{idle}]$ \Comment{Initial view of the workers.}
        \State $w' \gets [\texttt{idle}, \dots, \texttt{idle}]$ \Comment{Next view of the workers.}
        \State $timers \gets [0, 0, \dots, 0]$ \Comment{Initial values of the timers.}
        \State $\forall i\in\{\texttt{idle}, \texttt{active}, \dots\}: X_i \gets \{j\ |\ w[j] = i\}$\Comment{Partitions of worker states.}
        
        \Function{getPartition}{$i$}\label{line:protocol-get-partition}
            \Return $\{j\ |\ w[j] = i\}$
        \EndFunction

        \Procedure{Uncrashed}{$j$}
            \If{$w[j] \in \{ \texttt{active}, \texttt{waking-up}, \texttt{pending-timeout} \}$}\label{line:protocol-uncrashed-check-124}
                \State $\textsc{Send}(j, \texttt{wakeup});$\label{line:protocol-uncrashed-send-wakeup}
                \State $w[j] \gets \texttt{waking-up};\ timers[j] \gets 0;$
            \Else\label{line:protocol-uncrashed-check-03}
                \State $w[j] \gets \texttt{idle};\ timers[j] \gets 0;$
            \EndIf
        \EndProcedure
        
        \Procedure{Timeout}{$j$}
            \If{$w[j] \in \{ \texttt{active}, \texttt{pending-timeout} \}$}\label{line:protocol-check-current-state-1or4}
                \If{$\exists i: w[i] == \texttt{idle}$}\label{line:protocol-does-idle-worker-exist}
                    \State $\textsc{Send}(i, \texttt{wakeup});$\label{line:protocol-send-wakeup}
                    \State $w'[i] \gets \texttt{waking-up};\ timers[i] \gets 0;\ timers[j] \gets 0;$
                    \State \Return \texttt{snoozing}
                \Else{$\ timers[j] = 0;\ \Return\ \texttt{pending-timeout}$}\label{line:protocol-idle-worker-doesnt-exist}
                \EndIf
            \ElsIf{$w[j] == \texttt{waking-up}$}\label{line:protocol-check-current-state-2}
                    $timers[j] = 0;\ \Return\ \texttt{idle}$
                %\State $timers[j] \gets 0;\ \Return\ \texttt{active}$\label{line:protocol-wakeup-successful}
            \EndIf
        \EndProcedure
        
        \Procedure{Main}{}\label{algo:sieve:main}\Comment{Execution starts from here}
        \For{$i$ in $\{\texttt{idle}, \texttt{active}, \dots\}$}\label{line:protocol-main-get-partition}\Comment{Update the partitions.}
            \State $X_i \gets \textsc{getPartition}(i);$
        \EndFor
        \State $timers \gets \textsc{updateTimers}(timers);$\label{line:protocol-update-timers}
        \For{$j$ in $X_{\texttt{active}}$}\label{line:protocol-check-timeout-1}
            \If{$timers[j] == T$}
                \State $w'[j] \gets \textsc{Timeout}(j);$\label{line:protocol-call-timeout-1}
            \EndIf
        \EndFor
        \For{$j$ in $X_{\texttt{waking-up}}$}\label{line:protocol-check-timeout-2}
            \If{$timers[j] == d_{wakeup}$}
                \State $w'[j] \gets \textsc{Timeout}(j);$\label{line:protocol-call-timeout-2}
            \EndIf
        \EndFor
        \For{$j$ in $X_{\texttt{snoozing}}$}\label{line:protocol-check-timeout-3}
            \If{$timers[j] == d_{sleep}$}
                %\If{$M[j].type\ != \texttt{heartbeat} \lor M[j].time < \textsc{time}.\textsc{now}() - T$}\label{line:protocol-check-sleeping-last-message}
                    $\textsc{Send}(j, \texttt{sleep});$
                %\EndIf
                \State $timers[j] \gets 0;$\label{line:protocol-set-sleeping-timer-zero}
            \EndIf
        \EndFor
        \For{$j$ in $X_{\texttt{pending-timeout}}$}\label{line:protocol-check-timeout-4}
            \If{$timers[j] == T$}
                \State $w'[j] \gets \textsc{Timeout}(j);$\label{line:protocol-call-timeout-4}
            \EndIf
        \EndFor
        \State{$w \gets w';$}\label{line:protocol-update-state}\Comment{Update the worker states}
        \EndProcedure
    \end{algorithmic}
\end{algorithm}
}

\textbf{Step 4: Complete the proof.} For ease of exposition, we make an assumption:~$W_{noOverload}=W_{\min}$. %\af{Is this too risky?}.

The adversary for \mftcm{} crashes workers during a temporary \shock{}; consequently, all crashed workers must eventually be no longer crashed, though the duration of the \shock{} is unknown. After the \shock{}, the system begins in an arbitrary state that consists of ($i$) workers (that are either idle or active) and ($ii$) \mftcm{}’s view of the state of these  workers, with the corresponding timers.
%The adversary for \mftcm{} just crashes workers, and since the \shock{} is temporary, it has to uncrash all crashed workers at the end of the \shock{}; the duration of the \shock{} is unknown.
%The initial state for the system after the \shock{} includes workers that are either idle or active, \mftcm{}'s view of the workers, and corresponding timers.

Following the methodology in \S\ref{sec:methodology}, we must show that, from any such initial state:
($i$) \mftcm{} postpones \texttt{sleep} commands while the system is stabilizing;
($ii$) \mftcm{} issues enough \texttt{wakeup} commands to restore~$W = W_{min}$; and
($iii$) once~$W = W_{min}$, \mftcm{} ceases issuing \texttt{sleep} commands.
We prove them in reverse order.

($iii$): Because overload is an instantaneous function of~$W$, the moment~$W = W_{min}$, all active workers are no longer overloaded and respond to \mftcm{} with heartbeats. At that point, \mftcm{} knows that at least~$W_{min}$ workers are active and will never issue \texttt{sleep} events again.

($ii$): Let~$X_i$ denote the set of workers that \mftcm{} deems in state~$i$, for $i \in \{\texttt{idle},\texttt{active},\texttt{waking-up},\texttt{snoozing},\allowbreak \texttt{pending-timeout}\}$.
Define:~$Z = |X_{\texttt{active}}| + |X_{\texttt{snoozing}}| + |X_{\texttt{pending-timeout}}|$,~$V = |X_{\texttt{active}}| + |X_{\texttt{waking-up}}| + |X_{\texttt{snoozing}}| + |X_{\texttt{pending-timeout}}|$, and~$U = |X_{\texttt{active}}| + |X_{\texttt{waking-up}}| + |X_{\texttt{pending-timeout}}|$.
After the \shock{}, workers notify \mftcm{} upon recovering and are then placed in state \texttt{waking-up}.
If \mftcm{} deems a worker~$w_i$ in~$X_j$ for~$j \in \{\texttt{active},\texttt{snoozing},\texttt{pending-timeout}\}$, then~$w_i$’s actual state is \texttt{active}.
Thus, after the shock,~$Z = W$ holds invariantly.
Furthermore, executions satisfy the invariant~$U = W_{min}$. The invariant holds initially, because the system starts with~$W_{min}$ active workers and the rest idle.
Subsequently, workers in~$X_{\texttt{pending-timeout}}$ either remain or move to~$X_{\texttt{waking-up}}$; workers in~$X_{\texttt{waking-up}}$ either remain or move to~$X_{\texttt{active}}$; and whenever a worker leaves~$X_{\texttt{active}}$, it either moves to~$X_{\texttt{pending-timeout}}$ or to~$X_{\texttt{snoozing}}$ while another worker moves to~$X_{\texttt{waking-up}}$.
Therefore,~$U$ never changes, so~$U = W_{min}$ throughout. Since~$V \ge U = W_{min}$, we have~$|X_{\texttt{waking-up}}| \ge W_{min} - Z$.
Because~$Z = W$, it follows that~$|X_{\texttt{waking-up}}| \ge W_{min} - W$: there always are enough workers scheduled to wake up.

%To show the second point, let~$X_i$ be the set of workers deemed in state~$i$ by \mftcm{} for~$i\in\{0, 1, 2, 3, 4\}$.
%Let~$Z = |X_1| + |X_3| + |X_4|$,~$V = |X_1| + |X_2| + |X_3| + |X_4|$, and~$U = |X_1| + |X_2| + |X_4|$.
%After the \shock{}, since workers inform \mftcm{} upon uncrashing and \mftcm{} then deems them in state 2, if \mftcm{} deems worker~$w_i$ in~$X_j$ for~$j\in\{1, 3, 4\}$ then~$w_i$'s actual state is \texttt{active}.
%In other words, after the \shock{} we have~$Z = W$.
%Moreover, the executions of the system satisfy the invariant~$U = W_{min}$.
%To see why, note that the system starts from a state where there are~$W_{min}$ active workers with the rest idle, so initially we have~$U = |X_1| = W_{min}$.
%Furthermore, workers in~$X_4$ either stay there or move to~$X_2$, workers in~$X_2$ either stay there or move to~$X_1$, and whenever a worker leaves~$X_1$, it either moves to~$X_4$, or moves to~$X_3$ but another worker moves to~$X_2$.
%Therefore,~$U$ never changes, which means~$U = W_{min}$ all the time.
%Now note that~$V\geq U = W_{min}$, which implies~$|X_2|\geq W_{min} - Z$.
%Since after the \shock{} we have~$Z=W$, we get~$|X_2|\geq W_{min} - W$; there always are enough workers scheduled to wake up.

($i$): Note that at all times the difference between the wake-up times of any two workers in $X_{\texttt{waking-up}}$ is at most $T$.
This holds because \mftcm{} checks heartbeats every $T$ timesteps and retries workers in \texttt{pending-timeout} every $T$ timesteps.
Since $d_\texttt{wakeup} + T < d_\texttt{sleep}$, whenever a worker enters \texttt{snoozing}, its corresponding \texttt{sleep} event is scheduled only after all pending \texttt{wakeup} events have occurred.

We apply our methodology to the look-aside cache case study~\cite{bronson-2021-metastability, huang-2022-metastable} as well, and defer it to Section~\ref{sec:cache-incident} of supplementary materials for lack of space.

%% file: related.tex
\section{Related Work}\label{sec:related-work}

Metastable failures were first introduced by Bronson et al.~\cite{bronson-2021-metastability}, who informally characterize them as persistent, self-sustaining degraded modes of operation.
Huang et al.~\cite{huang-2022-metastable} expand on this work by presenting a broader suite of industrial metastability incidents, defining metastability in terms of persistent overload after a trigger, and categorizing self-sustaining mechanisms as either workload amplification or capacity degradation.
Habibi et al.~\cite{habibi-2023-msfmodel} focus on the retry storm incident, and leverage queuing theory and continuous-time Markov chains to explain how repeated request retries destabilize the system.
Isaacs et al.~\cite{analyzing-metastable-failures} reiterate metastable failures as self-sustaining congestive collapse, and present a host of tools ranging from a discrete-event simulator to production-level software to study the effects of various system parameters on emergent self-sustaining overload.
Farahbakhsh et al.~\cite{hotnets2025} provide an operational system model involving queues, and define metastability as never-ending difference between the outputs of two copies of a system: one starting after a \shock{}, and one starting from the normal initial state.
Anand~\cite{blueprint-anand2023} propose a toolchain to develop microservice-based systems in a modular, plug and play, and configurable fashion, and use their framework to demonstrate metastable failures in the deathstarbench suite~\cite{deathstarbench2019}.
While insightful, these works ($i$) do not characterize the underlying causes of metastable failures, and some ($ii$) cannot account for important metastability incidents that do not manifest as overload.

We have identified metastability incidents that prior studies have overlooked.
Floyd and Van Jacobson~\cite{floyd1993synchronization} model router synchronization using Markov chains, showing how message exchange can synchronize and overload the network.
Khan et al.~\cite{understanding-vicious-cycles-khan2011} describe a vicious cycle between a load balancer and a thermal manager, triggered by an AC failure.
Qian et al.~\cite{qian-2023-viciouscycles} report a Hadoop~\cite{hdfs} incident where a missed heartbeat causes recovery actions to interfere with heartbeat management, exposing harmful interactions between error handling and request handling.
Ford~\cite{ford2012icebergs} documents oscillations caused by faulty coupling between a load balancer and a power optimizer.

Distributed systems exhibit many failure modes,~$e.g.$, crash~\cite{FLP}, fail-stop~\cite{fail-stop}, fail-slow~\cite{Gunawi-fail-slow, limplock}, and Byzantine faults~\cite{Byz-generals}.
Metastability, however, is fundamentally different: it does not stem from permanent component failures but is an {\em emergent property} sustained by feedback loops spanning individually correct components.

Classic examples of failures caused by component interactions rather than local faults include deadlocks, livelocks, and race conditions.
Deadlocks and livelocks arise from cyclic dependencies among processes~\cite{Silberschatz}.
Deadlocks can be characterized using resource-allocation graphs~\cite{coffman-1971-deadlocks} and avoided via algorithms such as Banker’s algorithm~\cite{Dijkstra-banker} or prevented entirely using lock-ordering disciplines~\cite{havender68}.
Concurrent accesses to shared resources can lead to nondeterministic failures~\cite{Engler-racerx} and require atomicity.
Metastable faults also stem from cyclic dependencies, expressed through the concept of stabilization.
The solution, as in deadlock, lies in scheduling.

Self-stabilizing protocols were first introduced by Dijkstra~\cite{dijkstra-1974-selfstabilizing}.
Self-stabilizing systems converge to a set of good states and remain there.
Metastable faults and failures arise when such systems are composed.

Recent work has explored correct-by-design cluster management controllers.
Sun et al.~\cite{sun2024anvil} propose eventually stable reconciliation (ESR), requiring controllers to eventually drive the system to the desired state and keep it there.
They present Anvil, a tool for designing controllers and proving correctness using Verus~\cite{lattuada2023verus}.
While ESR resembles stabilization, their framework excludes metastability: it assumes components stop destabilizing the controller until it updates the state, whereas metastability occurs precisely when such problematic interactions persist.

%% file: conclusion.tex
\section{The Road Ahead}\label{sec:conclusion}

We introduce the first analytical characterization of metastable faults and failures, and present a methodology for designing MFT systems.
We apply our methodology to three case studies, showing the generality of our approach.
Our contributions set the stage for future work on building a toolbox to prevent, detect, and mitigate metastable failures.

% We posit that research on metastable faults and failures should be similar to security research. We are always at the mercy of adversarial \shock{}s that break implicit assumptions and reveal faults hidden across the stack.
% In fact, a similar, frustrating experience led researchers at Meta to adopt an adversarial view of their own system, managing finally to locate the fault~\cite{link-imbalance-bronson2014}.

%% file: appendix.tex
\section{Model and Formal Semantics}

This section includes our formal model, the semantics for the~$\rightsquigarrow^+$ operator, our definitions in the paper for ease of access, and the proof for Theorem 1 from the paper.

\subsection{Model}\label{sec:formal-model}
To reason about how components destabilize each other, we need a model that captures state changes, action effects, and how components read and write shared state.

A system~$S = (\Sigma_S, \mathcal{A}_S, \mathcal{V}, \mathcal{V}_E)$ is a state machine with a state space~$\Sigma_S$, a set of actions~$\mathcal{A}_S$, a set of variables~$\mathcal{V}$, and a set of variables~$\mathcal{V}_E$ such that~$\mathcal{V}_E\subseteq \mathcal{V}$---these are variables that the environment in which the system operates can modify.
Each state is a valuation of the variables in~$\mathcal{V}$.
~$S$ interacts with an environment that can modify the variables in~$\mathcal{V}_E$ by taking environment actions.
Environment actions are arbitrary---the system does not control the environment's behavior.
We denote state transitions of the system as~$s\rightarrow s'$ of the system, where~$s, s'\in\Sigma_S$.
We omit the subscript~$S$ when the system is clear from the context.
%We also omit~$\mathcal{V}$ for brevity, 

For a system~$S = (\Sigma, \mathcal{A}, \mathcal{V}, \mathcal{V}_E)$, a predicate~$P$ is a subset of~$\Sigma$; an environment predicate~$Q$ is a predicate that involves only variables in~$\mathcal{V}_E$. 
Any execution~$\sigma$ of the system is a trace~$\sigma=s_0\rightarrow s_1\rightarrow s_2\rightarrow\dots$, where~$s_i\in\Sigma$ for all~$i\geq 0$ and in each transition at least one of system and environment takes an action.
For every execution~$\sigma$ of the system, all suffixes of~$\sigma$ are also executions of the system.
Given some environment predicate~$E$ and system states~$s$ and~$s'$, a transition~$s\rightarrow s'$ is an~$E$-step if~$s$ and~$s'$ satisfy~$E$.
If during a transition~$s\rightarrow s'$ only the system takes an action, say~$\alpha$, we further label the transition as~$s\rightarrow^\alpha s'$.
%Whenever the environment has taken an action, followed by a system action, we say that the execution has taken a \emph{step}.
To keep notation short, we will omit mention of variables and write~$S=(\Sigma, \mathcal{A})$; it is to be inferred that the context implicitly indicates which variables the environment can modify.

If a system assigns a value to a variable following an action, we say the system \emph{writes to} the variable via the action, establishing a writes-to relation between them; if a system reads the value of a state variable with an action, we say that the system \emph{reads from} the variable via the action.

Let~$\bot$ denote the absence of an action for any system.
Given two systems~$S_1=(\Sigma_1, \mathcal{A}_1, \mathcal{V}_1, \mathcal{V}^1_E)$ and~$S_2=(\Sigma_2, \mathcal{A}_2, \mathcal{V}_2, \mathcal{V}^2_E)$, we define their composition as the system~$S_1\parallel_{\mathcal{V'}} S_2 = (\Sigma_1\times\Sigma_2, \mathcal{A}, \mathcal{V}_1\cup\mathcal{V}_2, \mathcal{V}')$, where~$\times$ denotes the Cartesian product and~$\mathcal{V}'\subseteq \mathcal{V}^1_E\cup\mathcal{V}^2_E$.
As for~$\mathcal{A}$, assuming~$\alpha\in\mathcal{A}_1$ and~$\beta\in\mathcal{A}_2$, each action~$\gamma\in\mathcal{A}$ is of one of  ($i$)~$(\alpha, \beta)$, ($ii$)~$(\alpha, \bot)$, and ($iii$)~$(\bot, \beta)$: the composition makes a state transition whenever at least one component takes an action.
The set~$\mathcal{V}'$ specifies the remaining environment variables in the composition, since components can partially be the environment for each other.
If, upon composition,~$S_1$ becomes responsible for writing to some variable of~$S_2$ that was previously written to by the environment, the environment stops writing to that variable in the composition.
To keep notation short, we will drop the subscript~$\mathcal{V}'$; the specifics of which component writes to which environment variable of another component are to be inferred from the context.
We generalize this to compositions of more than two systems, and denote the composition of the~$n$ systems~$\{S_i\}_{0\leq i < n}$ with~$\parallel_{0\leq i < n}S_i$.
Whenever the composition~$S$ takes an action~$\alpha=(\alpha_0,\dots, \alpha_{n-1})$ and the environment takes an action~$e$, the transition is serializable~\cite{papadimitriou1979serializability, bernstein1979serializability},~$i.e.$, the resulting state is equivalent to that resulting after \emph{some} serial execution of the actions~$\{\alpha_i\}_{0\leq i < n}$ and~$e$.

Given a composition of~$n$ systems~$\{S_i\}_{0\leq i < n}$, we lift the writes-to relation between systems and variables to a writes-to relation between systems.
If a system~$S_i$, via some action~$\alpha_i\in\mathcal{A}_i$, writes to a state variable of~$S_i$ that system~$S_j$ reads from, or directly to a state variable of~$S_j$ that~$S_j$ reads from, we say that~$S_i$ writes to~$S_j$ via~$\alpha_i$.
This relation induces a \emph{composition blueprint}: a directed graph whose vertices are systems and whose edges represent writes-to relations between them.
If~$S_i$ writes to~$S_j$, the edge is from~$S_i$ to~$S_j$.
For a system~$S_i$, we call the set of all systems that write to it its \emph{writing neighbors}, and denote it with~$W_i$.
Similarly, we call the set of all systems that~$S_i$ writes to as its \emph{reading neighbors}, and denote it with~$R_i$.

To facilitate our proofs, which entail liveness assertions, we assume that in every execution of a composition~$S$ ($i$) each component takes actions infinitely often, and ($ii$) during each transition at least one writing neighbor of every component takes an action.

\subsection{Semantics for~$\rightsquigarrow^+$}\label{sec:squiggly-semantics}

Let~$S=(\Sigma, \mathcal{A})$ be a system interacting with some environment,~$E\subseteq\Sigma$ be an environment predicate, and~$G\subseteq\Sigma$ a predicate.
A state~$s\in\Sigma$ satisfies a predicate~$P$,~$i.e.$,~$s\models P$, if~$s\in P$.
The temporal formula~$E\rightsquigarrow^+ G$ is an assertion on the executions of~$S$ as it is interacting with its environment.
Consider the execution~$\sigma=s_0\rightarrow s_1\rightarrow s_2\rightarrow\dots$, where~$s_i\in\Sigma$ for all~$i\geq 0$.
We define~$\sigma_{i:\infty}$ to be the execution~$s_i\rightarrow s_{i+1}\rightarrow\dots$.
~$\sigma$ satisfies~$E\rightsquigarrow^+ G$,~$i.e.$,~$\sigma\models E\rightsquigarrow^+ G$, if the following holds:
\begin{align*}
    \exists\ i: [\bigwedge_{0\leq j\leq i}(s_j\models E)]\Rightarrow [(s_i\models G) \land (\sigma_{i+1:\infty}\models G\ \mathcal{U}\ \neg E)],
\end{align*}
where~$\mathcal{U}$ denotes the \texttt{unless} temporal modality~\cite{manna1995temporal}.
The system~$S$ satisfies~$E\rightsquigarrow^+ G$,~$i.e.$,~$S\models E\rightsquigarrow^+ G$, if all of its executions satisfy~$E\rightsquigarrow^+ G$.

We use the following property of the~$\rightsquigarrow^+$ operator in the proof of Theorem 1 (\S\ref{sec:proofs}).

\begin{lemma}\label{lem:squiggly-tastes-funny}
    Let~$S=(\Sigma,\mathcal{A})$ be a system, and~$P$,~$Q$, and~$R$ predicates over~$\Sigma$.
    If~$S\models P\land Q\rightsquigarrow^+ R$ and~$S\models P\rightsquigarrow^+ Q$, then $S\models P\rightsquigarrow^+ Q\land R$.
\end{lemma}
\begin{proof}
    Consider some execution~$\sigma = s_0\rightarrow s_1\rightarrow s_2\rightarrow\dots$ of~$S$.
    Since~$S\models P\rightsquigarrow^+ Q$, there exists some~$i$ such that, if all~$s_j\models P$ for~$0\leq j\leq i$, then~$s_i\models Q$ and~$\sigma_{i+1:\infty}\models Q\ \mathcal{U}\ \neg P$.
    That is, after index~$i$,~$Q$ keeps holding in every state of~$\sigma$ unless~$P$ stops holding at some point.
    Similarly, since~$S\models P\land Q\rightsquigarrow^+ R$, therefore~$\sigma_{i+1:\infty}\models P\land Q\rightsquigarrow^+ R$---note that~$\sigma_{i+1:\infty}$ is a suffix of~$\sigma$, and is therefore an execution of~$S$.
    Thus, there exists some~$k$ such that, if all~$s_l\models P\land Q$ for~$i+1\leq l\leq k$, then~$s_k\models R$ and~$\sigma_{k+1:\infty}\models R\ \mathcal{U}\ \neg(P\land Q)$.
    Based on the above, if for all~$0\leq n\leq k$ we have~$\sigma_n \models P$, it holds that~$\sigma_{k}\models Q$ and~$\sigma_{k+1:\infty}\models Q\ \mathcal{U}\ \neg P$.
    Similarly, we infer that~$\sigma_{k}\models R$ and~$\sigma_{k+1:\infty}\models R\ \mathcal{U}\ \neg (P\land Q)$.
    Based on the properties of the \texttt{unless} operator, we infer that~$\sigma_{k+1:\infty}\models (Q\land R)\ \mathcal{U}\ \neg P$, and we independently establish~$\sigma_{k}\models Q\land R$.
    This proves that~$\sigma\models P\rightsquigarrow^+ Q\land R$.
    Since we picked~$\sigma$ arbitrarily, we thus have~$S\models P\rightsquigarrow^+ Q\land R$.
\end{proof}

\subsection{Faults and Failures}\label{sec:faults-failures-formal}

We provide our definitions for ease of accessibility when reading the proof.

\begin{definition}[\textbf{Potential function}]
    For a system~$S = (\Sigma, \mathcal{A})$ and a state predicate~$G\subseteq\Sigma$ representing a set of good states, a function~$f:\Sigma\rightarrow\mathbb{R}_{\geq 0}$ is a \emph{potential function} for~$(S, G)$ iff:
    \begin{enumerate}[label=P\arabic*, ref=P\arabic*, left=0pt]
        \item\label{P1} for all~$s\in\Sigma$,~$f(s) = 0\Leftrightarrow s\in G$; and

        \item\label{P2} for all~$s, s'\in\Sigma$ and~$\alpha\in\mathcal{A}$, if the system makes a transition~$s\rightarrow^\alpha s'$, then~$f(s')\leq f(s)$.
    \end{enumerate}
\end{definition}

\begin{definition}[\textbf{Stabilizing system}]
    For a system~$S=(\Sigma, \mathcal{A})$, a predicate~$G \subseteq \Sigma$, and a potential function~$f: \Sigma\rightarrow\mathbb{R}_{\geq 0}$ for the pair~$(S, G)$, the pair~$(S, f)$ is \emph{stabilizing} iff there exists an environment predicate~$E$ such that, as long as the environment takes actions such that the system state repeatedly satisfies~$E$, eventually the system state also satisfies~$f = 0$, and keeps satisfying~$f=0$ as long as environment actions keep maintaining~$E$.
\end{definition}

\begin{definition}[\textbf{Compatibility}]
    The stabilizing systems $(S_1, f_1)$,~$(S_2, f_2)$,~$\dots$, and~$(S_n, f_n)$ are \emph{compatible} if there exists an environment predicate~$E$ such that the following holds for~$\parallel_{0\leq i < n} S_i$, for all~$0\leq j < n$:
    \begin{enumerate}[label=C\arabic*, ref=C\arabic*, left=0pt]
        \item\label{C1}~$E\land (\bigwedge_{i\in W_j} f_i = 0)\rightsquigarrow^+ f_j = 0$. 
    
        % \item\label{C1}~$E\land f_1 = 0\rightsquigarrow^+ f_2 = 0$; and

        % \item\label{C2}~$E\land f_2 = 0\rightsquigarrow^+ f_1 = 0$.
    \end{enumerate}
    We call~$E$ a \emph{compatible environment} for the composition.
\end{definition}

\begin{definition}[\textbf{Destabilizing action}]
    Let~$(S_1, f_1)$ and~$(S_2, f_2)$ be two stabilizing systems in the composition~$S$ of some compatible stabilizing systems, where~$S_i = (\Sigma_i, \mathcal{A}_i)$ for~$i\in\{1, 2\}$.
    Let~$\alpha_1\in\mathcal{A}_1$ be an action with which~$S_1$ writes to~$S_2$.
    If there exists a compatible environment~$E$ for~$S$ and states~$s, s'\in\Sigma_2$ such that ($i$)~$s\rightarrow s'$ is an~$E$-step, and ($ii$) either~$f_2(s') \geq f_2(s) > 0$ or~$f_2(s') > f_2(s) = 0$, then we say that~$\alpha_1$ is destabilizing at~$s$.

    % Let~$E$ be a compatible environment for~$S$, and~$e$ be any environment action.
    % Let~$\alpha_1\in\mathcal{A}_1$ be an action with which~$S_1$ writes to~$S_2$, and~$s, s', s''\in\Sigma_2$ be such that we have~$s\rightarrow^e s'\rightarrow^{\alpha_1}s''$, and that~$s$,~$s'$, and~$s''$ all satisfy~$E$.
    % \begin{itemize}
    %     \item If~$f_2(s'') < f_2(s)$, we say that the pair~$(\alpha_1, e, E)$ is strictly stabilizing at~$s$.

    %     \item If~$f_2(s'')\geq f_2(s)$ ($f_2(s'') > f_2(s)$), we say that the tuple~$(\alpha_1, e)$ is destabilizing (strictly destabilizing) at~$s$.
    % \end{itemize}
\end{definition}

\begin{definition}[\textbf{Metastable fault}]
    Let~$i\oplus 1$ denote $i + 1\ \texttt{mod}\ k$ for~$0\leq i\leq k-1$ and some~$k$.
    Given compatible stabilizing systems~$(S_0, f_0)$,~$\dots$,~$(S_{n-1}, f_{n-1})$, where $S_i = (\Sigma_i, \mathcal{A}_i)$, their composition~$S = \parallel_{0\leq i < n}S_i$ \emph{has a metastable fault} iff there exists a cycle~$M = \{S_{i_0}, \dots, S_{i_{k-1}}\}$ of systems in the composition blueprint, and an action~$\alpha_{i_j}\in\mathcal{A}_{i_j}$ for~$0\leq j\leq k-1$, such that for all~$0\leq j < k$:
    \begin{enumerate}[label=M\arabic*, ref=M\arabic*, left=0pt]
        \item\label{M1} (\textbf{Writes-to}) each~$S_{i_j}$ writes to~$S_{i_{j\oplus 1}}$ via~$\alpha_{i_j}$; and

        \item\label{M2} (\textbf{Destabilization}) there exists a state~$s_{i_{j\oplus 1}}\in\Sigma_{i_{j\oplus 1}}$, such that~$f_{i_{j\oplus 1}}(s_{i_{j\oplus 1}}) > 0$ and~$\alpha_{i_j}$ is destabilizing at~$s_{i_{j\oplus 1}}$.
    \end{enumerate}
\end{definition}

\subsection{Proof of Theorem 1}\label{sec:proofs}

%\section{Proofs}\label{sec:proofs}

We begin by proving two lemmas.

\begin{lemma}\label{lemma:compatible}
    Let~$(S_0, f_0)$, $(S_1, f_1)$,~$\dots$, and~$(S_{n-1}, f_{n-1})$ be compatible stabilizing systems,~$S$ be their composition, and~$J$ be a subset of~$\{0,\dots,n-1\}$.
    If~$E$ is a compatible environment for~$S$, then~$E\land(\bigwedge_{j\in J}f_j = 0)$ is a compatible environment for the systems with indices in~$[n]\backslash J$.
\end{lemma}
\begin{proof}
    Note that, for all~$i$:
    \begin{align*}
        E&\land (\bigwedge_{k\in W_i} f_k = 0)\rightsquigarrow^+ f_i = 0\equiv\\
        (E&\land(\bigwedge_{k\in W_i\cap J}f_k = 0))\land(\bigwedge_{k\in W_i \setminus J}f_k = 0)\rightsquigarrow^+ f_i = 0.
    \end{align*}
    Moreover, note that:
    \begin{align*}
        E\land(\bigwedge_{k\in J}f_k = 0) \Rightarrow E&\land(\bigwedge_{k\in W_i\cap J}f_k = 0).
    \end{align*}
    If we consider the systems with indices in~$J$ as part of the environment for the remaining systems, then the set of writing neighbors of any remaining component~$S_i$ changes from~$W_i$ to~$W_i \setminus J$.
    We conclude that~$E\land(\bigwedge_{j\in J}f_j = 0)$ is a compatible environment for the systems with indices not in~$J$.
\end{proof}

\begin{lemma}\label{lem:e-step}
    Let~$(S_0, f_0)$,~$(S_1, f_1)$,~$\dots$, and~$(S_{n-1}, f_{n-1})$ be compatible stabilizing systems, where~$S_i=(\Sigma_i,\ \mathcal{A}_i)$ for~$0\leq i < n$, and~$E$ be a compatible environment for~$S = \parallel_{0\leq i < n}S_i$.
    Consider a transition~$s\rightarrow s'$ of~$S$ wherein the environment takes action~$e$ and the writing neighbors of some component~$S_i$ take actions~$\{\alpha_j\}_{j\in W_i}$, where~$\alpha_j\in\mathcal{A}_j$, and assume that the serialization order is~$s=s_0\rightarrow^{\beta_1} s_1\rightarrow\dots\rightarrow s_k\rightarrow^{\beta_{k+1}} s_{k+1}\rightarrow\dots\rightarrow^{\beta_{m}}s_{m}=s'$, where~$\beta_{k+1}=e$ and the rest of~$\{\beta_l\}$ are a permutation of the actions~$\{\alpha_j\}_{j \in W_i}$.
    If~$s\rightarrow s'$ is an~$E$-step, then every~$s_{l}\rightarrow s_{l+1}$ is also an~$E$-step for~$0\leq l\leq |W_i|-1$.
\end{lemma}
\begin{proof}
    The only transition during which the environment variables of the composition change is the transition~$s_k\rightarrow^e s_{k+1}$.
    Therefore, since~$s=s_0$ and~$s'=s_m$ satisfy~$E$, and since the actions~$\{\alpha_j\}_{j\in W_i}$, and therefore~$\{\beta_l\}$, do not change the environment variables, we conclude that all states~$\{s_l\}$ satisfy~$E$.
    We conclude that every transition~$s_l\rightarrow s_{l+1}$ is an~$E$-step.
\end{proof}

\faultsNecessary*
% \begin{theorem}\label{thm:crisis}
%     Let~$(S_0, f_0)$, $(S_1, f_1)$,~$\dots$, and~$(S_{n-1}, f_{n-1})$ be compatible stabilizing systems, and~$S=\bigwedge_{0\leq i < n}S_i$ their composition.
%     If~$S$ has a metastable failure, then it has a metastable fault.
% \end{theorem}
\begin{proof}
    We prove this theorem by proving its contrapositive: if~$S$ has no metastable faults, then~$S$ and any compatible environment~$E$ satisfy~$E\rightsquigarrow^+\bigwedge_{0\leq i < n}f_i = 0$, which implies~$\Box E\Rightarrow\Diamond\Box\bigwedge_{0\leq i < n}f_i = 0$,~$i.e.$,~$S$ does not have a metastable failure.
    %if~$S$ has no metastable faults, if all systems are normal, and if all systems are locally compatible with their reading neighbors, then~$S$ and any compatible environment~$E$ satisfy~$E\rightsquigarrow^+\bigwedge_{0\leq i < n}f_i = 0$,~$i.e.$,~$S$ does not have a metastable failure.
    We proceed by induction on~$n$.
    \begin{description}
        \item[Base case.] If~$n=1$, then we have a single component~$(S_0, f_0)$.
        Stabilization for~$S$ implies the existence of an environment predicate~$E$ such that~$E\rightsquigarrow^+ f_0 = 0$; therefore, a compatible environment for~$S_0$ exists, and for any such environment~$E$, based on compatibility we have~$E\rightsquigarrow^+ f_0 = 0$ as~$S_0$ has no writing neighbors.
        This is exactly what we want to prove.

        \item[Hypothesis.] Any compatible environment~$E'$ for the composition~$S' = \parallel_{0\leq i < k} S'_i$ of any~$k < n$ compatible stabilizing systems~$(S'_0, f'_0)$, $(S'_1, f'_1)$,~$\dots$, and $(S'_{k-1}, f'_{k-1})$, that has no metastable faults, in tandem with~$S'$, satisfies~$E'\rightsquigarrow^+\bigwedge_{0\leq i < k}f'_i = 0$.

        \item[Step.] Let~$E$ be a compatible environment for~$S$.
        Since~$S$ has no metastable faults, then there exists some system~$S_j$ such that either~$W_j = \emptyset$---Let~$J'$ be the index set of all such systems---or for any set of actions~$\{\alpha_k\}_{k\in W_j}$ with~$\alpha_k\in\mathcal{A}_k$ for~$k\in W_j$, the actions~$\{\alpha_k\}$ are not destabilizing at~$s$, which means that each action can only decrease~$f_j$ if positive or maintain it at zero during any~$E$-step; let~$J''$ be the index set of components with the latter quality.
        For every~$j'\in J'$, compatibility implies~$E\rightsquigarrow^+ f_{j'} = 0$.
        Now pick one~$j''\in J''$.
        Since~$f_{j''}$ is a potential function for~$S_{j''}$, then~$S_{j''}$'s own actions do not increase~$f_{j''}$.
        Now, consider an~$E$-step~$s\rightarrow s'$ of~$S$ where the writing neighbors of~$S_{j''}$ take actions~$\{\alpha_l\}_{l\in W_{j''}}$---at least one such component takes an action in each transition per our model.
        This transition is equivalent to some serial order, and according to Lemma~\ref{lem:e-step}, all the corresponding intermediate transitions are also~$E$-steps.
        Therefore, since each action~$\alpha_l$ for~$l\in W_{j''}$ either decreases~$f_{j''}$ if positive or maintains it at zero during an~$E$-step, we conclude that either~$f_{j''}(s) > f_{j''}(s')$ or~$f_{j''}(s) = f_{j''}(s') = 0$.
        %Now, ($i$) the actions of~$S_{j''}$'s neighbors can only either decrease~$f_{j''}$ or maintain it at zero in any~$E$-step, ($ii$) any joint action of its neighbors is serializable, ($iii$) if a transition via the actions of~$S_{j''}$'s neighbors is an~$E$-step then the intermediate transitions in any serialization of this transition are also~$E$-steps, and ($iv$) each neighbor of~$S_{j''}$ will take actions infinitely often.
        Therefore,~$f_{j''}$ will eventually decrease to 0 and remain there,~$i.e.$,~$E\rightsquigarrow^+ f_{j''} = 0$.
        Letting~$J = J'\cup J''$, we have shown that~$E\rightsquigarrow^+ \bigwedge_{j\in J} f_j = 0$.

        Consider now the systems with indices not in~$J$.
        Based on Lemma~\ref{lemma:compatible},~$E\land(\bigwedge_{j \in J} f_j = 0)$ is a compatible environment for the composition of these systems.
        Moreover, their composition inherits not having a metastable fault from~$S$, as otherwise~$S$ would have a metastable fault.
        Therefore, based on the induction hypothesis, we have~$E\land(\bigwedge_{j \in J} f_j = 0)\rightsquigarrow^+ \bigwedge_{j\in[n]\backslash J} f_j = 0$.
        Since we also have~$E\rightsquigarrow^+\bigwedge_{j\in J}f_j = 0$, based on Lemma~\ref{lem:squiggly-tastes-funny} we deduce:
        \begin{align*}
            E\rightsquigarrow^+ (\bigwedge_{j\in J}f_j = 0) \land (\bigwedge_{j\in [n] \setminus J}f_j = 0),
        \end{align*}
        which is just~$E\rightsquigarrow^+ \bigwedge_{0\leq j < n} f_j = 0$.
        This finishes our proof.

    \end{description}
\end{proof}

\section{The Look-Aside Cache Incident}\label{sec:cache-incident}

Following the original implementation demonstrating this failure~\cite{look-aside-cache-estyak2022, huang-2022-metastable}, consider a cache, a webserver, and a database.
The webserver receives client requests and forwards them to the cache.
If the result is a hit, the webserver responds to the client.
If the cache query returns with a miss, the webserver forwards the request to the database, and also starts a timer.
If it receives a response from the database before the timer expires, it sends the response to the client \emph{and} updates the cache with the result; if, on the other hand, the result takes too long to come back, it drops the request.

The cache, if full, helps respond to a lot of the client requests, and therefore moderates significantly the load on the database.
If a \shock{} ($e.g.$, crash) leaves the cache empty, then suddenly all of the client requests will be forwarded to the database, severely congesting it.
As a result, requests will eventually timeout at the webserver, which means that it will not update the cache by virtue of dropping the requests.
This behavior stops the cache from adequately warming up after the \shock{}, sustaining therefore the congestion.

In the following, we first apply our characterization to this incident to reveal the faults and the failure.
We will then demonstrate that, for this particular incident, removing the fault is not particularly out of reach, and show how simple design decisions help the designer get rid of the fault.
Finally, we apply the methodology from \S\ref{sec:methodology} to render the system MFT without eliminating the faults.

\textbf{Potential Functions.} The potential function for the cache is~$f_1 = C_{max} - C$, the number of keys for which it does not have a value---$C_{max}$ is the maximum number of keys the cache can store values for.
It stabilizes in the presence of an environment that populates it with new keys.
The potential function~$f_2$ for the webserver is the number of pending requests, and it stabilizes in the presence of an environment that gives it enough acknowledgments and not too many new requests.
The potential function for the database is~$f_3 = \max\{0, Q-s\}$, where~$Q$ is the number of requests in its queue and~$s$ is its service rate.
It stabilizes in the presence of an environment that does not overload it.

\textbf{Compatibility.} Let~$E$ denote an environment that, at each step, gives to the webserver~$r$ client requests, where~$0 < r < s$.
The database's only writing neighbor is the webserver.
Because the client load does not overload the database, eventually the database will eliminate any existing congestion and respond to the webserver in a timely fashion, regardless of the webserver's potential.
Therefore, we have~$E\land f_2 = 0\rightsquigarrow^+ f_3 = 0$.
The webserver's writing neighbors are the cache and the database.
If the cache is full, and the database is not congested, the webserver will get timely responses for all of the client requests it receives,~$i.e.$, we have~$E\land f_1 = 0\land f_3 = 0\rightsquigarrow^+ f_2 = 0$.
Finally, the cache's only writing neighbor is the webserver.
If we repeatedly have~$f_2 = 0$, it means that the webserver has no pending requests, implying that the database is not congested; the cache will be populated, and we have~$E\land f_2 = 0\rightsquigarrow^+ f_1 = 0$.
We have established compatibility.

\textbf{Metastable Fault.} We model the cache with the single action~$\texttt{cache-serve}$, with which it responds to the requests sent from the webserver in a step.
The cache might be empty, and its response to all of the requests can be a miss, therefore this action can possibly maintain the webserver's potential~$f_2$.
We model the webserver with one action as well: the action~$\texttt{webserver-serve}$.
This action sends new requests to the cache, missed ones to the database, updates the cache, and drops also the requests that time out.
This action might send the cache zero updates if all pending requests time out, and therefore maintains the cache's potential~$f_1$; there exists a metastable fault between the cache and the webserver.

The action~$\texttt{webserver-serve}$ can increase the database's potential~$f_3$ if it sends the database more than~$s$ client requests, and the database's only action~$\texttt{database-serve}$, which serves client requests in its queue and responds to the webserver, might not decrease the webserver's potential if all of the corresponding requests have timed out.
Therefore, a metastable fault exists between the webserver and the database as well.

\textbf{Metastable Failure.} The two metastable faults are fanned into a metastable failure by two scheduling decisions: ($i$) the webserver schedules to drop requests every~$T$ steps for some~$T$, and ($ii$) the database does not prioritize new requests over older ones by simply serving requests from its queue according to arrival order.
After the \shock{} that empties the cache, the webserver congests the database by repeatedly sending it more client requests than it can handle.
Once this condition persists long enough, old requests will time out at the webserver.
Then, since the backlog in the database's queue has only been increasing, \emph{all younger requests} will also eventually time out as they will wait even longer than the first request that timed out.
The result is a stable backlog that starves the cache, and therefore sustains itself indefinitely.

\subsection{Removing the Faults}
Unlike the other case studies, where the fault is irremovable (retry storm) or too delicate to locate a priori (oscillating membership), in this incident formalizing the faults reveals a fix to remove them altogether.
The webserver should randomly tag some of the pending requests among each~$s$ requests that it sends to the database, and \emph{never time them out}.
This way, assuming a suitable distribution of keys in the incoming client requests, the webserver will always update the cache after a \shock{}, and therefore its action will always decrease the cache's potential~$f_1$.
This breaks the cyclic destabilizing interaction between the cache and the webserver.
The fixed system will never experience the metastable failure outlined above.

\subsection{Applying the MFT Methodology}

One can render the system MFT without removing the fault.

\textbf{Derive metastability skeleton.} The original implementation showcasing this failure~\cite{look-aside-cache-estyak2022} requires a careful design including a memcached cache~\cite{memcached}, an NGINX~\cite{nginx} webserver expressed in PHP~\cite{php}, and a mySQL database~\cite{mysql}.
We instead express each component as a \dsl{} agent, and abstract away every detail that is not relevant to the metastable failure.
What remains is the metastability skeleton of the composition, which captures how requests flow in the system and how each agent responds to different types of requests.
We ignore key replacement in the cache, as it does not pertain to the warmup phase and has little effect on this particular failure.

\textbf{Study dynamics.} Once equipped with the metastability skeleton of the system, we can probe the system's behavior by asking \dsl{} to generate the vector field.
For simplicity, we only look at the vector field of cache and database potentials. 
Figure~\ref{fig:look-aside-cache} demonstrates the result, where the client sends requests at a rate of~$r=40$ requests per step, the database serves them at a rate of~$s=10$ requests per step, and the webserver timeout is~$T=8$ steps.
We instruct \dsl{} to generate 100 traces, using two simultaneous \shock{}s with a random duration: cache crash and surge in client demand.

\begin{figure}[H]
    \centering
    \includegraphics[width=\columnwidth]{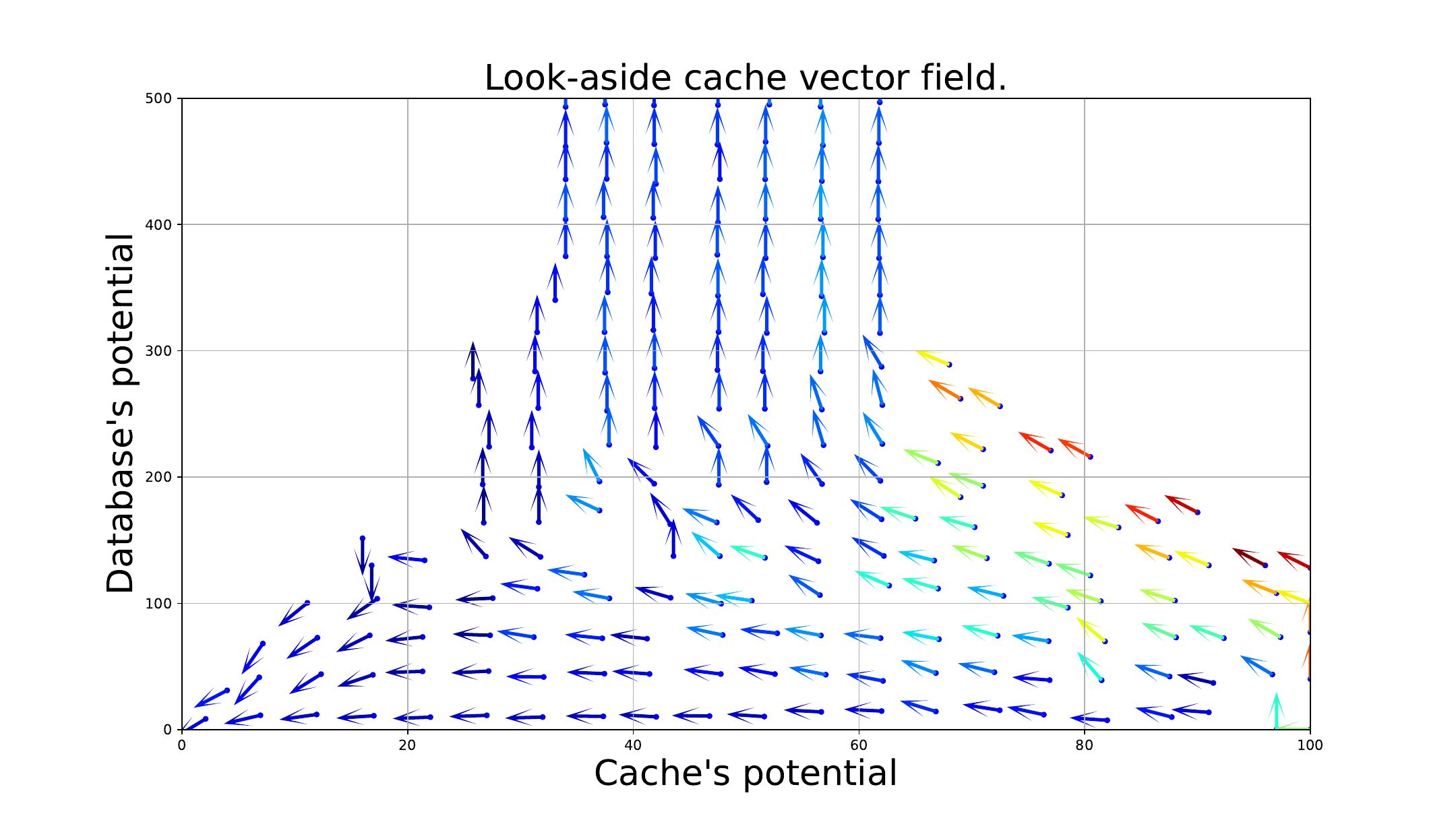}
    \caption{The \dsl{} vector field of the look-aside cache. Arrow color indicates tendency intensity at that point; warmer means higher intensity.}
    \label{fig:look-aside-cache}
\end{figure}

We make two observations: ($i$) the figure demonstrates the two coexisting tendencies, one pulling towards a full cache and an empty database, and the other pulling towards a stagnant cache and an ever-increasing database backlog, and ($ii$) the backlog is stable; once the bad tendency wins the tug of war, there is no hope for the system to recover without manual intervention.
This reinforces our account of the failure, wherein all pending requests eventually time out once the first timeout happens.

\textbf{Manage scheduling.} To make this system MFT, we have to change how the components schedule their actions.
One option is to pick a~$T$ large enough that the cache has ample time to warm up.
Another, more subtle approach is to change the scheduling in the database, and enforce a mechanism for prioritizing younger requests.
This reduces the chances of them timing out, and gives the system more time to fill the cache.
The details of the prioritization depend on the parameters and the key distribution in client requests.

\textbf{Complete the proof.} Starting from any arbitrary state in the aftermath of a \shock{} emptying the cache, a race starts between the tendency filling in the cache and the timer ticking in the webserver.
If the cache is sufficiently populated before the first timeout in the webserver, the system will never experience a metastable failure.
On the other hand, if the first timeout happens before the cache is sufficiently populated, then the cache will never be sufficiently populated.
Therefore, proving correctness for this system entails proving that the cache fills fast enough for there to be no timeouts, which corresponds to proving R2.
As for R1, since, as explained above, being sufficiently full is a stable property, the system will never timeout, and therefore never destabilize itself.